\documentclass[12pt,a4,english,leqno]{article}
\usepackage[T1]{fontenc}
\usepackage[latin9]{inputenc}
\usepackage{geometry}
\usepackage{verbatim}
\usepackage{amsmath}
\usepackage{graphicx}
\usepackage{longtable}
\usepackage{setspace}
\usepackage[round,sort,authoryear]{natbib}\setcitestyle{round}
\usepackage{amsfonts}
\usepackage{graphicx}
\usepackage{babel}
\usepackage{amsthm}
\usepackage[toc,page]{appendix}
\usepackage{booktabs,caption}
\usepackage{enumitem}
\usepackage{amsfonts}
\usepackage{accents}
\usepackage{comment}
\usepackage{natbib}
\usepackage{bibentry}
\usepackage{mathtools}
\usepackage{upgreek}
\usepackage[colorlinks=true,linkcolor=blue, citecolor=blue, urlcolor = blue]{hyperref}
\usepackage{lmodern}
\usepackage{avant}
\usepackage{pgfplots}

\usepackage[T1]{fontenc}
\usepackage{titlesec}
\titlelabel{\thetitle.\quad}

\newcommand\norm[1]{\left\lVert#1\right\rVert}

\usepackage{fancyhdr}
\makeatletter
\newcommand{\nextverbatimspread}[1]{%
  \def\verbatim@font{%
    \linespread{#1}\normalfont\ttfamily
    \gdef\verbatim@font{\normalfont\ttfamily}}
}
\makeatother

\titleformat*{\section}{\large \bfseries}
\titleformat*{\subsection}{\normalsize \bfseries}
\titleformat*{\subsubsection}{\small \bfseries}

\providecommand{\U}[1]{\protect\rule{.1in}{.1in}}
\newtheorem{theorem}{Theorem}

\newtheorem{corollary}{Corollary}

\newtheorem{example}{Example}

\newtheorem{lemma}{Lemma}

\newtheorem{proposition}{Proposition}
\newtheorem{remark}{Remark}

\newtheorem{assumption}{Assumption}

\usepackage{mathtools}

\DeclarePairedDelimiter{\floor}{\lfloor}{\rfloor}

\numberwithin{equation}{section}

\AtBeginDocument{}

\newif\ifshow 
\showfalse 

\ifshow
  \includecomment{wrap}
\else
  \excludecomment{wrap} 
\fi

\begin{document}

\pagenumbering{roman}





\title{ {\normalsize \textbf{BOOTSTRAPPING NONSTATIONARY AUTOREGRESSIVE PROCESSES  \\ \vspace{-0.5\baselineskip} WITH PREDICTIVE REGRESSION MODELS\thanks{Article history: November 2020, February 2022, July 2023.
\ 
\textcolor{blue}{\textit{MSC2020 classifications:}} 62F05, 62F12, 62M10.   
\textit{Keywords:} nonstationary autoregression, IVX filtration, bootstrap asymptotics, mixed gaussian distribution.   
}
}
}
\\
\author{ \small BY CHRISTIS KATSOURIS\footnote{Lecturer in Economics, Department of Economics, University of Exeter Business School, Exeter EX4 4PU, United Kingdom. \textit{E-mail Address}: \textcolor{blue}{c.katsouris@exeter.ac.uk}. }  
\\
\textit{University of Southampton and University of Exeter} 
      }
}

\date{}

\maketitle


\begin{abstract}
\vspace{-2.9\baselineskip}
We establish the asymptotic validity of the bootstrap-based IVX estimator proposed by \cite{phillipsmagdal2009econometric} for the  predictive regression model parameter based on a local-to-unity specification of the autoregressive coefficient which covers both nearly nonstationary and nearly stationary processes. A mixed Gaussian limit distribution is obtained for the bootstrap-based IVX estimator. The statistical validity of the theoretical results are illustrated by Monte Carlo experiments for various statistical inference problems.   



\end{abstract}



\setcounter{page}{1}
\pagenumbering{arabic}

\section{Introduction.}  

Consider the first-order autoregressive process $\left\{ X_t \right\}_{t=1}^{+ \infty}$, defined by the following recursive process
\begin{align}
\label{zero.model}
X_t = \vartheta_n X_{t-1} + \varepsilon_t, \ \ \ \ X_0 = 0, 
\end{align}   
where $\left\{ \varepsilon_t \right\}$ are independent $\mathcal{N} (0,1)$ random sequences. The least squares estimator $\hat{ \vartheta}_n$ of $\vartheta$, based on a sample of $n$ observations $\left\{ X_1,...,X_n \right\}$ is given by the following expression
\begin{align}
\hat{\vartheta}_n = \left( \sum_{t=1}^n X^2_{t-1} \right)^{-1} \left( \sum_{t=1}^n X_{t-1} X_t \right).
\end{align}
A well known asymptotic result is that $\hat{\vartheta}_n$ is a consistent estimator for $\vartheta_n$ for all values of $\vartheta_n \in \left( - \infty, + \infty  \right)$. Specifically, the asymptotic distribution of $\hat{\vartheta}_n$ depends on restrictions imposed on the admissible parameter space of $\vartheta$. In particular, in the stable case, which implies that the parameter space of $\vartheta$ takes values within the unit circle, that is, $| \vartheta | < 1$, various seminal studies such as  \cite{mann1943statistical} and \cite{anderson1959asymptotic} among others, have shown that
\begin{align}
\sqrt{n} \left( \hat{\vartheta}_n -  \vartheta_n \right) \overset{ d }{ \to } \mathcal{N} \big( 0, 1 - \vartheta^2 \big) \ \ \ \text{as} \ \ n \to \infty.
\end{align}   
When the true parameter lies on the boundary of the parameter space, such that $| \vartheta | = 1$, it has been shown by  \cite{white1958limiting} and \cite{phillips1987towards} that the following limiting distribution holds 
\begin{align}
n \left( \hat{\vartheta}_n -  \vartheta_n \right) \overset{ d }{ \to }  \left( \int_0^1 W(r)^2 ds \right)^{-1} \left( \int_0^1 W(r) d W(r) \right) \ \ \ \text{as} \ \ n \to \infty,
\end{align}

\newpage 

where $W(r)$ for some $0 \leq r \leq 1$ is a standard Wiener process within the probability space $\left( \Omega, \mathbb{P}, \mathcal{F}_t \right)$. On the other hand, when the true parameter of the autoregressive process given by expression \eqref{zero.model} is outside the unit circle (explosive parameter region), such that $| \vartheta | > 1$, it has been shown by \cite{white1958limiting} and \cite{anderson1959asymptotic} that the following asymptotic result holds 
\begin{align}
\sqrt{n} \left( \hat{\vartheta}_n -  \vartheta_n \right) \overset{ d }{ \to } \text{Cauchy} \left( 0, \vartheta^2 - 1 \right) \ \ \ \text{as} \ \ n \to \infty.
\end{align}
The particular limiting distribution discontinuity with respect to the parameter space of the autoregressive model of order one has been extensively discussed in the literature (see, \cite{bercu2001large}, \cite{jeganathan1991asymptotic},  \cite{phillips2007limit}, \cite{proia2020moderate}, \cite{benke2021nearly}). Without loss of generality, considering for each $n \in \mathbb{N}$, the statistical experiment $\mathcal{E}_n$ corresponding to the observations $\left\{ X_0, ..., X_n \right\}$, then the sequence $\left( \mathcal{E}_n \right)_{ n \in \mathbb{N} }$ is locally asymptotically normal (\textit{LAN}) when $| \vartheta | < 1$, it is locally asymptotically Brownian functional (\textit{LABF}) when $| \vartheta | = 1$, and it is locally asymptotically  mixed normal (\textit{LAMN}) when $| \vartheta | > 1$ (see, \cite{le1986asymptotic}, \cite{le2000asymptotics} and \cite{phillips1989partially}). Thus, statistical inference procedures can crucially depend on the true value of the population parameter $\theta_n$ in certain regions of the parameter space.   

Motivated by these observations various studies have focus on the development of estimation and inference methodologies for nonstationary time series models (see, \cite{rao1978asymptotic}, \cite{dickey1979distribution}, \cite{chan1987asymptotic1, chan1988limiting},  \cite{phillips1987towards, phillips1987time, phillips1988regression, phillips1988weak}, \cite{phillips1988testing}, \cite{abadir1993limiting},  \cite{larsson1995asymptotic}, \cite{buchmann2007asymptotic} and \cite{phillips2007limit} and \cite{horvath2003bootstrap}). Specifically, the class of nearly nonstationary time series models as proposed by the seminal studies of \cite{phillips1987towards, phillips1987time} who developed asymptotics for near-integrated time series provide a unified approach to inference. In practise, expressing the autoregressive coefficient with the local-unit-root specification, such that, $\vartheta_n = \big( 1 + c / n \big)$, where $n$ the sample size and $c$ denotes the nuisance parameter of persistence, allows to encompass such moderate deviations from the unit boundary (e.g., unit root, near-integrated or explosive, see \cite{stoykov2019least}). Our main research interest is to investigate the finite and large sample properties of bootstrap-based estimators and test statistics (e.g., \cite{moreno2012unit}) for predictive regression models when regressors are generated as nonstationary autoregressive processes.

Predictive regression models are commonly used for modelling macroeconomic data which exhibit high persistence (see, \cite{zhu2014predictive}). Furthermore, the nonstandard nature of the inference problem in predictive regressions has been previously discussed in the literature such as by the studies of \cite{Phillips1990statistical},  \cite{cavanagh1995inference} and \cite{jansson2006optimal}, with the main challenge being the robustness of inference methodologies to the nuisance parameter of persistence (see, \cite{astill2023bonferroni})). A unified framework for robust estimation and inference regardless of the abstract degree of persistence has been proposed by \cite{phillipsmagdal2009econometric} and \cite{kostakis2015Robust} (see, also \cite{chen2013uniform}, \cite{breitung2015instrumental}, \cite{li2017uniform, liu2019unified}, \cite{yang2021unified}  and \cite{jayetileke2021predictive}). We argue that the implementation of resampling methodologies for nonstationary time series models provides a way for conducting inference, when non-pivotal quantities are involved, which can ensure accurate asymptotic approximations (see more recently \cite{demetrescu2022extensions}). Thus, establishing the asymptotic validity (in the spirit of \cite{cavaliere2020inference}) of these bootstrapped approximations is crucial in other testing problems as well such as when considering structural break (see, \cite{georgiev2018testing}, \cite{katsouris2022partial, katsouris2023testing, katsouris2023structural} and \cite{fei2023ivx}).

\newpage 

A commonly used resampling methodology used to study the distributional properties of a statistic of interest is based on the well known Efron's bootstrap method as in \cite{efron1986bootstrap}. Various studies consider the asymptotic validity of the bootstrap-based estimator in \eqref{zero.model}. More specifically, the validity of explosive autoregressive processes is studied by \cite{basawa1989bootstrapping, basawa1991bootstrapping} while in general asymptotics for the bootstrap in stationary models is predominately concerned with independent observations (see, \cite{freedman1981bootstrapping, freedman1984bootstrapping}). Although it has been initially documented in the literature that the bootstrap cannot work for dependent processes (see, \cite{bose1988edgeworth}) it was anticipated that it would work if the dependence is taken care of while resampling (see, \cite{politis1994stationary} and \cite{paparoditis2003residual, paparoditis2005bootstrapping}). Therefore, the dependence structure of innovations found in predictive regression models  has to be correctly implemented when considering a bootstrap procedure. 

Our framework considers that the dependence structure is preserved due to the joint Gaussianity assumption imposed on the error sequences of the predictive regression model along with the nonstationary autoregressive process. In particular, motivated by the limit theory of cointegrating regressions the bootstrap implementation is adjusted to capture the dependence induced by the linear process representation of the innovation sequence (see, \cite{li1997bootstrapping}, \cite{li2001bootstrapping}, \cite{psaradakis2001bootstrap}, \cite{psaradakis2001bootstrapURs} and \cite{inoue2002bootstrapping}). Furthermore, the optimization method employed can affect the limiting distribution of related test statistics. The main quantity of interest is the risk of the normalized error, such that, $r_n \left( \hat{\theta}_n - \theta \right)$ of an estimator $\hat{\theta}_n$ with a suitable slow varying rate function $r_n$ where $r_n \to \infty$ as $n \to \infty$. The exact form of the particular rate function depends on the properties of the estimator under consideration; such as $\sqrt{n}-$consistency is only ensured for the OLS estimator while a different rate of convergence holds for the IVX estimator. Consequently, deriving probability bounds for the specific risk quantity using the Berry-Essen approach, as well as for its bootstrapped counterpart, in order to assess the efficiency of finite-sample approximations depend on the chosen estimation methodology\footnote{Related statistical theory is discussed in the books of  \cite{tanaka2017time} and \cite{hamilton2020time} as well as in the studies of  \cite{hall2014martingale} and  \cite{billingsley2013convergence}. Moreover, \cite{bickel1981some} presents some asymptotic theory aspects for the bootstrap while \cite{mykland1992asymptotic} proposes relevant asymptotic expansions.}.

In general, by the classical CLT, the distribution of a pivotal quantity 
\begin{align}
Q_n = \sqrt{n} \big( \mu_n - \mu \big) / s_n
\end{align}
tends weakly to $\mathcal{N}(0,1)$. Let $F_n$ be the empirical distribution of $X_1,..., X_n$, putting mass $1/n$ on each $X_t$. When one considers the bootstrap resampling, the $n$ data points $X_1,..., X_n$ are treated as a population, with distribution function $F_n$ and mean $\mu_n$, while $\mu_n^*$ is considered as an estimator of $\mu_n$. In practise, the idea is that the behaviour of the quantity $Q_n^*$ mimics that of $Q_n$. Thus, the distribution of $Q_n^*$ could be computed from the data and used to approximate the unknown sampling distribution of $Q_n$ (see, \cite{praestgaard1993exchangeably}). 

This paper develops limit theory for the asymptotic validity of the bootstrapped based IVX estimator in predictive regression models that satisfies the regularity conditions of valid instrumental based estimators such as relevance and orthogonality while preserving the robust properties to the nuisance parameter of persistence. To the best of our knowledge this is a novel implementation and further establishes the robust properties of the IVX estimator for bootstrap-based inference.

\newpage 

\subsection{Preliminary Theory.}

Concurrent with our work, and independently, \cite{georgiev2021extensions} also develop asymptotic results for the bootstrap test based on the IVX estimator who focus on predictability tests based on the sub-sample approach. Another key point here is that a comparison of the asymptotic behaviour of various bootstrapping procedures especially for the case of predictive regression models with regressors of abstract degree of persistence will be essential to obtain insights regarding the accuracy of bootstrap approximations. Notice that the IVX estimator is a consistent estimator of the predictive regression model coefficient although it has a different convergence rate in comparison to the OLS estimator. Therefore, the bootstrap which is well known to deliver asymptotic refinements over first-order asymptotic approximations as well as bias corrections can be applied to solve the asymptotic bias problem (see \cite{cavaliere2022bootstrap} for the counterargument). On the other hand, any additional bias term that appears in the asymptotic expansion reduces the precision of the test. 

Consider the partial-sum process below:  
\begin{align}
W_n(t) = \frac{1}{ \sqrt{n} } \sum_{i=1}^{ \floor{nt} } Z_i^n,    
\end{align}
then it holds that 
\begin{align}
X_n(t) = X_n(0) + \int_0^{ \floor{nt} / n } b \big( X_n(s) \big) ds + \int_0^t \sqrt{ a \big( X_n(s - ) \big) } dW_n(s) + error.   
\end{align}
where $W$ is the standard Brownian motion. 

Moreover this convergence suggests that $X_n$ should converge to a solution of the limiting stochastic differential equation. Thus, a key step in the application of the stochastic differential equation approach is to show that the sequence of stochastic integrals in the approximating equation converges to the corresponding stochastic integral in the limit\footnote{A specific example is the case in which one considers the random weighted bootstrap. In this case it is more interesting to study the Marcinkiewicz-Zygmund strong law for randomly weighted sums of dependent random variables under dependent random designs. }. Let $D$ be the space of cadlag functions $f: [0,1] \to \mathbb{R}$ equipped with the Skorokhod topology. As a measurable structure on $D$ we consider the corresponding Borel $\sigma-$algebra $\mathcal{B}(D)$. We write $X_n \overset{d}{\to} X$ whenever $X_n$, X are random variables taking values in $\big( D, \mathcal{B} (D) \big)$ such that $X_n$ converges weakly to $X$ in $D$ as $n \to \infty$, that is, $\mathsf{lim}_{ n \to \infty } \mathbb{E} \big[ f (X_n) \big] \to \mathbb{E} \big[ f(X) \big]$ for all bounded continuous functions $f: D \to \mathbb{R}$.  

We have that 
\begin{align}
\frac{1}{a_n} \sum_{k=1}^{ \floor{n . } } \varepsilon_k \overset{d}{\to} J, \ \ \ \text{as} \ \ n \to \infty.    
\end{align}
Therefore, it holds that
\begin{align}
\mathbb{P} \big( J(t) = 0 \ \forall t \in [0,1] \big) = 0.    
\end{align}

\newpage 

We define the Ornstein-Uhlenbeck process $J_c = \big( J_c(t) \big)_{ t \in [0,1] }$ driven by $J$ by setting the following stochastic expression 
\begin{align}
J_c(t) = J(t) - c \int_0^t e^{ - c (t-s) } J(s) ds, \ \ \ t \in [0,1], c \in \mathbb{R}.   
\end{align}



The usefulness of the bootstrap resampling method is exactly to provide a way to ensure uniform size that is robust to the presence of effects (such as conditional heteroscedasticity) that can affect sufficient control of the type of errors when conducting hypothesis testing. Moreover, exact calculation of the critical value of the test statistic is approximated based on Monte Carlo simulations. Due to the computational burden of the bootstrap procedure we use relative small sample sizes. However, by taking a sufficiently large number of Monte Carlo samples the approximation error can be made arbitrary small thus ensuring that the bootstrap approximation will work well given regularity conditions hold that we discuss in more details below.  

Consider a test statistic $\mathcal{T}_n$ of a parameter estimator such that $\mathcal{T}_n = \mathsf{\eta}(n) \left( \hat{\theta}_n - \theta_0 \right)$. Let $\mathcal{T}^{*}_n$ denote the bootstrap version of $\mathcal{T}_n$ such that $\mathcal{T}^{*}_n = \mathsf{\eta}(n) \left( \hat{\theta}^{*}_n - \theta_0 \right)$. Let $\hat{F}_n (u) := \mathbb{P}^{*} \left( \mathcal{T}^{*}_n \leq u \right), u \in \mathbb{R}$ denote the distribution function conditional on the original pair of data.The bootstrap $\mathsf{p-value}$ is defined as below: 
\begin{align}
\hat{p}_n := \mathbb{P}^{*} \left( \mathcal{T}^{*}_n \leq \mathcal{T}_n \right)  = \hat{F}_n ( \mathcal{T}_n )  
\end{align}
First-order asymptotic validity of $\hat{p}^{*}_n$ requires that $\hat{p}_n$ converges in distribution to a standard uniform distribution such that $\hat{p}_n \overset{d}{\to} \mathsf{Unif} [0,1]$.

\paragraph{Minimax Optimality} A complementary aim of this paper is to show that the sufficient conditions for the error bounds are indeed necessary in some applications. Let us define a test $\phi$, which is a Borel measurable map, $\phi: \mathcal{X}_n : \mapsto \left\{ 0, 1 \right\}$. Then, for a class of null distributions $\mathcal{P}_0$, we denote the set of all level $\alpha$ tests by  
\begin{align}
\Phi_{n, \alpha} := \left\{  \phi: \underset{ P \in \mathcal{P}_0 }{ \mathsf{sup} } \ \mathbb{P}_P^{(n)} \left( \phi = 1 \right) \leq \alpha \right\}.   
\end{align}

We focus on the minimax framework to evaluate the performance of test statistics and model estimators which enables us to study the fundamental limits of hypothesis testing while ensuring a (strong) uniform guarantee over a large class of (null and alternative) distributions. The notion of minimax performance is widely used to quantify the difficulty of a statistical problem. More specifically, in the hypothesis testing literature, it is also common to study the performance of tests against fixed or directional alternatives such as cases in which $\beta_n = \beta_0 + \frac{\delta}{\sqrt{n}}$ for some $\delta > 0$. For example, the problem of severe empirical size distortions when $\rho$ is large it has been also documented in the case of structural break tests in dynamic econometric models (see, \cite{kuan1994implementing}). Therefore, this implies that we are likely to obtain false rejection of constant coefficients in dynamic models when variables have high persistence and endogeneity. 

\newpage 

\subsection{Outline of the paper.}

The rest of the paper is organized as follows. In Section \ref{Section2}, we introduce the predictive regression model and model assumptions. We motivate the implementation of the IV based estimator which is robust to the nuisance parameters of persistence. In Section \ref{Section3}, we prove the asymptotic validity of the bootstrap IVX estimator and verify that we obtain an identical limit distribution as in the standard case. Section \ref{Section4}, provides a Monte Carlo Simulation study. In particular, the asymptotic properties of the different bootstrap proposals under the null and under the alternative are investigated and their relative power performance is analyzed. 
Section \ref{Section5}, concludes. 

\subsection{Notation.}

Throughout the paper, all limits are taken as $n \to \infty$, where $n$ is the sample size. Then, the symbol $"\Rightarrow"$ is used to denote the weak convergence of the associated probability measures as $n \to \infty$ (as defined in \cite{billingsley2013convergence}) which implies convergence of sequences of random vectors of continuous c\`adl\`ag functions on $\mathcal{D} \left( [0,1] \right)$ equipped with the Skorokhod topology. The symbol $\overset{d}{\to}$ denotes convergence in distribution and $\overset{p}{\to}$ denotes convergence in probability within a suitable probability space $( \Omega, \mathcal{F}_t, \mathbb{P} )$. Finally, $\mathcal{O}(a_n)$ denotes the rate of convergence of estimators, such that if $\theta_n = \mathcal{O}(a_n)$, implies that, $\theta_n / a_n$ is bounded. Also we denote with $\mathbb{P}^{*}$, $\mathbb{E}^{*}$ and $\mathsf{Var}^{*}$ the probability measure, expected value and variance operators induced by the bootstrap resampling procedure conditional on the time series observations of the original sample.

\section{The model, main assumptions and estimators.}
\label{Section2}

In this section we introduce the model under consideration in this paper, the main modelling assumptions as well as two different estimation methodologies proposed in the literature. 

\subsection{Predictive regression model.}

Consider the predictive regression model given by 
\begin{align}
\label{predictive1}
Y_t &= \beta_n X_{t-1} + u_{t}, \ \ t \in \left\{ 1,...,n \right\},  \\
\label{predictive2}
X_t &= \rho_n X_{t-1} + v_t, \ \ \text{with} \ \ X_0 = 0.
\end{align}

Then, the regressor $X_t$ of the predictive regression model is assumed to be generated via the following process
\begin{align}
X_t = \varrho_n X_{t-1} + v_t, \ \ \ X_0 = 0, \ \ \varrho_n = \left( 1 + \frac{ c }{ n^{\gamma} } \right), \ \ \ t \in \left\{ 1,...,n \right\},
\end{align} 
where $\upbeta_n \in \mathbb{R}$ is the (model) parameter of interest while the autoregressive coefficient $\varrho_n  \in \mathbb{R}$ is defined such that $\varrho_n = \big( 1 +  \frac{ \displaystyle c }{ \displaystyle n^{\gamma} } \big)$ with $c  \in ( - \infty , + \infty)$ and $\gamma \in (0,1]$.

\newpage 

Notice that $v_t$ is assumed to be a zero mean, stationary and ergodic process with finite autocovariances given by $\tilde{\sigma} := \mathbb{E} \left( v_t v_{t-j} \right)$ and $\omega^2 = \sum_{j = - \infty}^{ \infty } \tilde{\sigma}$ is finite and nonzero. Furthermore, the nuisance coefficients $c$ and $\gamma$ are considered as tuning parameters which determine the degree of persistence for the regressors and therefore without loss of generality the particular novel feature proposed by \cite{phillips1987towards, phillips1987time} provides a tractable way to consider the nature of possible nonstatioary regressors when estimating predictive regressions. In terms of distributional assumptions, we assume that the predictive regression model is generated from a bivariate Gaussian distribution, such that $e_t = \left( u_t, v_t \right)^{\prime} \overset{ \textit{i.i.d} }{ \sim } \ \mathcal{N} \big( 0, \boldsymbol{\Sigma} \big)$ with $t \in \left\{ 1,...,n \right\}$. Further related theoretical results which apply to the innovation sequence of the system can be found in \cite{phillips1992asymptotics}. 

\begin{assumption}
\label{assumption1}
Let $\left\{ \epsilon_t \right\}_{ t \in \mathbb{Z} }$ be a sequence of \textit{i.i.d} continuous random variables with $\mathbb{E} \left( \epsilon_1 \right) = 0$ and $\mathbb{E} \left( \epsilon^2_1 \right) = 1$, and with a characteristic function $\phi(\tau)$. Then, the regressor $X_t$ of the predictive regression model is assumed to be generated via the following process
\begin{align}
X_t = \varrho_n X_{t-1} + v_t, \ \ \ X_0 = 0, \ \ \varrho_n = \left( 1 + \frac{ c }{ n^{\gamma} } \right), \ \ \ t \in \left\{ 1,...,n \right\},
\end{align} 
with $\gamma = 1$, where $\kappa$ is a real constant such that $c \in ( - \infty, 0)$  and $v_t = \sum_{j=0}^{\infty} c_j \epsilon_{t-j}$, with $\sum_{j=0}^{ \infty } c_j \neq 0$ and the                                          summability condition $\sum_{j=0}^{ \infty } j^{1+ \delta} | c_j | < \infty$ holds for some $\delta > 0$.
\end{assumption}  

\begin{assumption}
\label{assumption2}
\textit{(i)} Let $\left\{ u_t, \mathcal{F}_{t} \right\}_{ t \geq 1 }$, where $\mathcal{F}_{t}$ is a sequence of increasing $\sigma-$fields which is independent of $\epsilon_j$, $j > t $, forms a martingale difference sequence satisfying  $\mathbb{E} \big( u^2_{t} | \mathcal{F}_{t-1} \big) \to \sigma^2 > 0$, as $t \to \infty$ and $\underset{ t \geq 1 }{ \mathsf{sup} } \ \mathbb{E} \big( | u_{t}|^4 \big| \mathcal{F}_{t-1} \big) < \infty$.

\textit{(ii)} $X_t$ is adapted to $\mathcal{F}_t$, and there exists a correlated Brownian motion $\big( W, V \big)$ such that the following joint weakly convergence holds
\begin{align}
\left( \frac{1}{\sqrt{n} } \sum_{j=1}^{[nr]} \epsilon_{j-1} , \frac{1}{\sqrt{n} } \sum_{j=1}^{ [nr] } u_{j}   \right) \Rightarrow \big( W(r), V(r) \big)
\end{align}
on $\mathcal{D}_{ \mathbb{R}^2 } \left( [0,1] \right)^2$ as $n \to \infty$.
\end{assumption}
Define the partial sum process $V_n(r)$ and $U_n(r)$ as below
\begin{align}
\big( V_n(r), U_n(r) \big) = \frac{1}{ \sqrt{n} } \sum_{t=1}^{ \floor{nr} } \big( v_t, u_t \big), \ \ \ \ r \in [0,1].
\end{align}

\begin{remark}
Assumption \ref{assumption1} and \ref{assumption2} provide necessary conditions for the development of the asymptotic theory (similarly in \cite{Wang2012specification}) for the corresponding bootstrap IV based estimator we are aiming to examine its statistical validity. In terms of the estimation procedure for the parameter of the predictive regression model we consider both the least squares estimation as well as the endogenous  instrumentation method proposed by \cite{phillipsmagdal2009econometric}.     
\end{remark}

\newpage 

\subsection{Main Results}

We are interested to examine the asymptotic behaviour for the $\hat{\upbeta}_n$ estimator, in the case of both least square estimation as well as the two-step endogenous instrumental variable estimation approach followed for constructing the IVX filtration. To derive the asymptotic distribution of both estimators we consider the Ornstein-Uhlenbeck (OU) process 
\begin{align}
\label{OUprocess}
J_{c}(r) = \int_{0}^r e^{(r-s)c} dW(s), \ \ \ \ \ \text{with} \  0 \leq r \leq 1 \ \text{and} \ c \in (0,+\infty). 
\end{align}
which satisfies the stochastic differential equation $d J_{c}(r) = c J_{c}(r) dr + dW(r)$. Then, the following asymptotic results hold, based on the local to unity limit law proposed by \cite{phillips1987time,  phillips1987towards} such that $\frac{X_{[nr]}}{\sqrt{n} } \Rightarrow J_{c}(r)$, where $J_{c}(r)$ is the OU process defined in expression \eqref{OUprocess}.
\begin{align}
\frac{1}{n \sqrt{n} } \sum_{t=1}^{[nr] } X_t    & \Rightarrow \int_0^r J_{c}(s) ds, \ \ \ \ \ \ \ \ \ \text{with} \ 0 \leq s \leq r, \\
\frac{1}{n } \sum_{t=1}^{[nr] } X_t u_t  & \Rightarrow \int_0^r J_{c}(s) dW(s), \ \text{with} \ 0 \leq s \leq r. 
\end{align}
\begin{remark}
The convergence of the sample moments to their stochastic integral counterparts are proved in the seminal study of  \cite{phillips1988regression} who also derived invariance principles for the multivariate model. In our study we consider the corresponding univariate limit results since we consider the univariate predictive regression model with persistent regressors. Furthermore, we consider the bootstrap invariance principle to obtain the weak convergence of the sample moments for the bootstrapped-based test statistics. 
\end{remark}
The ordinary least squares estimate of $\upbeta_n$ denoted with $\hat{\upbeta}_n^{ols}$ is given by 
\begin{align}
\hat{\upbeta}_n^{ols} = \left( \sum_{t=1}^n X^2_{t-1} \right)^{-1} \left( \sum_{t=1}^n X_{t-1} Y_t \right)
\end{align} 
Therefore, since $\sqrt{n} \left( \hat{\upbeta}_n^{ols} - \upbeta  \right) = \mathcal{O}_p( \frac{1}{\sqrt{n}} )$ then it converges in distribution which implies 
\begin{align}
\label{limit1}
n \left( \hat{\upbeta}_n^{ols} - \upbeta  \right) 
&= \left( n^{-2} \sum_{t=1}^n X_{t-1}^2  \right)^{-1}  \left( n^{-1} \sum_{t=1}^n X_{t-1} u_t \right)  
\nonumber
\\
&\Rightarrow \left( \int_{0}^1 J^2_{c}(s) ds \right)^{-1} \left(  \int_{0}^1 J_{c}(s) dW(s) \right)
\end{align}

\medskip

\begin{remark}
The limit distribution given by \eqref{limit1} shows that the existence of the nuisance parameter $c$ in the underline stochastic process, when the parameter $\upbeta_n$ is obtained based on the least squares estimation weakly convergence to a nonstandard limit. On the other hand this is not the case for the corresponding IVX estimator as we demonstrate in the next section. 
\end{remark}

\newpage 

\subsection{IVX estimator}

In this section we discuss the main methodology employed for the estimating the nonstationary time series model. Specifically, the IVX instrumentation proposed by \cite{phillipsmagdal2009econometric} provides important  advantages rather than using the conventional differencing approach prior to fitting the model. Although employing the difference operator is a common approach for ensuring covariance stationarity in time series observations due to the fact that it can remove trends and seasonality it can also contribute to loss of information such as the long-memory properties\footnote{In particular, \cite{kasparis2015nonparametric} consider how the interplay of the long-memory properties of innovations and the persistence properties of regressors in a nonparametric predictive regression affect the limiting distributions of test statistics and corresponding model estimators.} of these observations which are eliminated when differencing is applied. Therefore, the IVX estimator in contrast to its OLS counterpart is found to be robust to the persistence properties of the regressors as captured by the autoregressive specification without altering the persistence properties of regressors. The particular estimation approach can accommodate both nearly nonstationary and nearly stationary processes, that is, generalizing this way the integration order of regressors for the predictive regression. 

Consider the instrumental variable $Z_{tn}$ which is based on the IVX methodology of \cite{phillipsmagdal2009econometric}. Then, the IVX instrumental variable is constructed  
\begin{align}
Z_{tn} = \sum_{j=0}^{t-1} \left( 1 + \frac{c_z}{n^{\upgamma_z}}  \right)^j \big( X_{t-j} - X_{t-j-1} \big), \ \ \ \ \text{where} \ c_z \in (- \infty, 0), \upgamma_z \in (0,1).
\end{align}  
Denote with $\hat{\upbeta}_n^{ivx}$ the corresponding IVX based estimator for the parameter of the predictive regression given by the system of equations in expressions \eqref{predictive1} and \eqref{predictive2}. Then, standard IV estimation arguments provide an expression for the IVX estimator 
\begin{align}
\hat{\upbeta}_n^{ivx} =  \left( \sum_{t=1}^n Z_{t-1} X_{t-1} \right)^{-1} \left( \sum_{t=1}^n Z_{t-1} Y_t  \right)
\end{align}
Thus, we consider the limit distribution of the random variable $\psi_n$ as defined below 
\begin{align}
\psi_n :=  n^{ \frac{ 1 + \upgamma_z }{2} } \left( \hat{\upbeta}_n^{ivx} - \upbeta_n \right) \equiv \left( \frac{1}{  n^{ 1 + \upgamma_x } } \sum_{t=1}^n Z_{t-1} X_{t-1} \right)^{-1} \left( \frac{1}{  n^{ \frac{1 + \upgamma_x }{2}  } }  \sum_{t=1}^n Z_{t-1} u_t  \right)
\end{align}

\begin{lemma}
\label{lemma1.main}
Under Assumption \ref{assumption1} and \ref{assumption2}, the following invariance principle holds
\begin{align}
\sum_{t=1}^{ [nr] } Z_{t-1} X_{t-1} \Rightarrow - \frac{1}{c_z} \left( r \omega_{xx}^2 +  \int_0^r J_{c}(s) dJ_{c}(s) \right), \ \ \text{with} \ \ r \in [0,1]. 
\end{align}
\end{lemma}

\medskip

\begin{remark}
Related weakly convergence arguments which can be employed for the proof of Lemma \ref{lemma1.main} can be found in \cite{phillipsmagdal2009econometric}. The proof of the particular result consists of a key limit result for deriving the asymptotic distribution of the bootstrap based IVX estimator which is the main goal of this paper. 
\end{remark}

\newpage 

A simple implementation of the limit given by Lemma \ref{lemma1.main} after appropriate normalization shows that the following nuisance-free limiting distribution holds
\begin{align}
\psi_n :=  n^{ ( 1 + \upgamma_z ) / 2 } \left( \hat{\upbeta}_n^{ivx} - \upbeta_n \right) \Rightarrow \psi \equiv \mathcal{MN} \left( 0, \tilde{ \boldsymbol{\Sigma} } \right)
\end{align}
where $\psi$ follows a mixed Gaussian distribution. The mixed Gaussian property of the $\psi_n$ random variate, is instrumental in developing inference procedures robust to abstract degree of persistence (see, \cite{phillips2013predictive, phillips2016robust}). 

Therefore, it remains to determine the components of the stochastic covariance matrix $\tilde{ \boldsymbol{\Sigma} }$. More precisely, to determine the covariance matrix of the random variable $\psi$, we consider the martingale difference sequence $\xi_{nt} :=  \left( \frac{1}{ n^{ ( 1 + \upgamma_z ) / 2 } } Z_{t-1} u_t , \frac{1}{ \sqrt{n} } v_t  \right)^{ \prime }$. Then the martingale conditional variance of $\xi_{nt}$ has the following structure 
\begin{align}
\label{mds.covariance1}
\sum_{t=1}^n \mathbb{E} \big( \xi_{nt} \xi_{nt}^{\prime} | \mathcal{F}_{nt-1} \big) =
\begin{bmatrix}
\displaystyle \left( \frac{1}{ n^{ 1 + \upgamma_z } } \sum_{t=1}^n Z_{t-1}^2 \right) \sigma^2_{uu} \ \ & \ \ \displaystyle \left( \frac{1}{ n^{1 + \frac{\upgamma_z}{2}}} \sum_{t=1}^n Z_{t-1} \right) \sigma_{uv} 
\\
\\
\displaystyle \left( \frac{1}{ n^{1 + \frac{\upgamma_z}{2}}} \sum_{t=1}^n Z_{t-1} \right) \sigma_{vu} \ \  & \ \ \displaystyle \sigma^2_{vv}
\end{bmatrix}
\end{align}

\begin{remark}
Note that as shown by \cite{phillips2007limit} $n^{ - \left( 1 + \upgamma_z / 2 \right)} \sum_{t=1}^n z_{t-1} = o_p(1)$, which implies that the martingale conditional covariance \eqref{mds.covariance1} has off-diagonal elements which converge in probability to zero. Furthermore, note that without loss of generality, we operate under the assumption that $\upgamma_x = 1$ and $c \in (- \infty, 0)$. Relaxing the parameter space of these two quantities allows to consider regressors generated from general nonstationary autoregressive processes, covering both near nonstationary and near stationary processes. Nevertheless, the robust property of the IVX estimator implies that regardless of the persistence properties of regressors asymptotically converges to a pivotal mixed Gaussian distribution. Therefore, the particular large sample property result to conventional statistical inference when considering Wald type statistics after suitable normalization (see,  \cite{phillips2013predictive, phillips2016robust}).     
\end{remark} 
Notice that it has been shown that when the residual process $u_t$ is correlated with the regressor $x_t$, the limiting distribution of the $t-$ratio is no longer normal (see, \cite{phillips1986multiple}).  As a remedy to the particular bias, the fully modified regression proposed by \cite{Phillips1990statistical} as an alternative way of constructing the test statistic as we present later. Overall a key component of our framework is that it establishes the invariance principle which concerns the weak convergence (within the Shorokhod topology) of the bootstrap partial sum process to Brownian motion. Then, based on the continuous mapping theorem, it can be used to obtain asymptotic distributions of various bootstrapped statistics without making parametric assumptions on the underlying model. In particular this approach is build upon the strong approximation and the Beveridge-Nelson representation of linear processes. The invariance principles for the $\textit{i.i.d}$ innovation and the bootstrapped counterpart are first developed using the strong approximation of the partial sum process by the standard Brownian motion, and subsequently the invariance principles for the general linear process with $\textit{i.i.d}$ innovations and the corresponding bootstrapped process are established by their Beveridge-Nelson representations as in \cite{phillips1992asymptotics}.   

\newpage 

Let the partial sum process of $( \varepsilon_t )$ be defined by 
\begin{align}
W_n(t) = \frac{1}{ \sigma \sqrt{n} } \sum_{j=1}^{ \floor{nt} } \varepsilon_j.
\end{align}
Then, we have that $W_n \overset{ d }{ \to } W$ in the space of $\mathcal{D} [0,1]$ of cadlag functions, where $W$ is the standard Brownian motion. The space $\mathcal{D} [0,1]$ is equipped with the Skorokhod topology. 

In particular for the bootstrap, we first obtain or estimate $\left\{ \varepsilon_t \right\}_{t=1}^n$ from the sample of size $n$ and get $\left\{ \hat{\varepsilon}_t \right\}_{t=1}^n$ . then, we resample from the empirical distribution of $\left\{ \hat{\varepsilon}_t \right\}_{t=1}^n$, that is, the distribution with point probability mass $1 /n$ based on the time series observations of size $n$, to get the bootstrap sample $\left\{ \varepsilon_t^{*} \right\}_{t=1}^n$ as the \textit{i.i.d} samples from the empirical distribution of $\left\{ \hat{\varepsilon}_t \right\}_{t=1}^n$. Both $( \hat{\varepsilon}_t  )$ and $(\varepsilon_t^{*})$ are dependent upon the sample size $n$, and we may more precisely denote them as triangular arrays $( \hat{\varepsilon}_{nt} )$ and $(\varepsilon_{nt}^{*})$.  

Furthermore, certain studies have previously suggested the inconsistency of the bootstrap estimator when the parameter is on the  boundary of the parameter space (see, \cite{andrews2000inconsistency}). However, notice within our setting the particular consideration has a different interpretation as such a phenomenon occurs via the local unit root process. In other words, the local-to-unity specification provides a natural bound on these type of processes when modelling dependent time series and therefore the asymptotic behaviour of estimators based on predictive regression models are analytically tractable. Therefore, when considering bootstrapping nonstationary autoregressive processes using predictive regression models, in practise we are mainly concern about deriving an identical weakly convergence argument similar to the non-bootstrapped  conventional limit result. 

Our main goal is to develop a consistent bootstrap estimator of the distribution of the stochastic quantity $\psi_n :=  n^{ \frac{ 1 + \upgamma_z }{2} } \left( \hat{\upbeta}_n^{ivx} - \upbeta_n \right)$, regardless of the persistence properties of the regressors of the predictive regression model as captured by the autoregressive process with a LUR coefficient. More precisely, under the assumption that the regressors are generated by the nonstationary process $X_t = \varrho_n X_{t-1} + v_t$ with $\varrho_n = \left( 1 - \frac{ \kappa }{ n^{\upgamma_x} } \right)$ in Assumption \ref{assumption1}, then we aim to investigate whether $\psi_n$ and $\psi_n^{\star} :=  n^{ \frac{ 1 + \upgamma_z }{2} } \left(  \hat{\upbeta}_n^{\star ivx} - \upbeta_n \right)$ converge to the same distribution limit within a suitable probability space. To do this, we consider all features of the distribution of $\psi_n$ such as the covariance matrix in \eqref{mds.covariance1}.

\section{The Bootstrap and Hypothesis testing.}
\label{Section3}

A commonly used approach to evaluate the limiting distributions in nonstationary time series models based on asymptotic approximations relies on finite samples bootstrapped samples. The exact type of the bootstrap procedure follwed to replicate the dependence structure of the model can affect both the validity of the results as well as the accuracy of those approximations. In this section, we study the asymptotic behaviour of the bootstrap based IVX estimator for different bootstrap schemes. To do this, we impose additional assumptions for the associated partial sum processes. Let $\mathcal{D}$ be the space of cadlag functions, $f : [0,1] \to \mathbb{R}$ equipped with the Shorokhod topology (e.g., see \cite{csorgHo2003donsker}).

\newpage

\paragraph{Notation} Denote with $\mathcal{Z}_{nt}^{*}$ a sequence of bootstrap statistics, then the following modes of convergence hold (see related definitions in  \cite{goncalves2015Bootstrap})
\begin{itemize}
\item $\mathcal{Z}_{t}^{*} = o_{ P^{*} }(1)$ in probability, or $\mathcal{Z}_{t}^{*} \overset{ P^{*} }{ \to } 0$ in probability, if for any 
\begin{align}
\epsilon > 0, \delta > 0 \ \ \underset{ n \to \infty }{ \mathsf{lim} } \mathbb{P} \big[ \mathbb{P}^{*} \big( \left| \mathcal{Z}_{n}^{*} \right| > \delta \big) > \epsilon \big] = 0.
\end{align}
\item $\mathcal{Z}_{t}^{*} = \mathcal{O}_{ P^{*} }(1)$ in probability, if for all $\epsilon > 0$ there exists a $M_{\epsilon} < \infty$ such that 
\begin{align}
\underset{ n \to \infty }{ \mathsf{lim} } \mathbb{P} \big[ \mathbb{P}^{*} \big( \left| \mathcal{Z}_{n}^{*} \right| > M_{\epsilon} \big) > \epsilon \big] = 0.
\end{align} 
\item $\mathcal{Z}_{n}^{*} \overset{ d^{*} }{ \to } \mathcal{Z}$ in probability if, conditional on the sample, $\mathcal{Z}_{n}^{*}$, weakly converges to $\mathcal{Z}$ under $P^{*}$, for all samples contained in a set with probability converging to one. Specifically, we write $\mathcal{Z}_{n}^{*} \overset{ d^{*} }{ \to } \mathcal{Z}$ in probability if and only if $\mathbb{E}^{*} \big( f \left( \mathcal{Z}_{n}^{*} \right) \big) \to \mathbb{E} \left( f(\mathcal{Z}) \right)$ in probability for any bounded and uniformly continuous function $f$.  
\end{itemize}
Therefore, for the asymptotic validity of the bootstrap algorithm one needs to consider similar arguments as when deriving the limiting distribution of nonstationary time series models. In particular, \cite{muller2011efficient} discusses various aspects of asymptotic efficiency with respect of the weak convergence of functionals to their brownian motion counterparts. 

\subsection{Bootstrap Procedure.}

Similar investigation in the literature with respect to the bootstrap procedure include the study of \cite{davidson2008wild} who consider a framework for the Wild bootstrap although not in a nonstationary time series regression environment. In this paper we apply the Wild bootstrap algorithm using the OLS as well as the IVX residuals. Another example include the study of \cite{phillips1990asymptotic} who develop asymptotic theory for residual-based tests of cointegration. Next, we describe the estimation procedure for the bootstrap based IVX estimator, which can help us shed light on the specific asymptotic expressions and properties we need to consider in order to derive its limiting distribution. 

Consider again the predictive regression with multiple predictors. Simulate the following DGP 
\begin{align}
Y_t &= \upbeta_n X_{t-1} + u_{t}, \ \ t \in \left\{ 1,...,n \right\},  
\\
X_t &= \varrho_n X_{t-1} + v_t, \ \ \text{with} \ \ X_0 = 0.
\end{align}
where the parameters of interest are $\upbeta_n$ and $\varrho_n$.  Denote with $u_t = \left( u_{t},  v_{t} \right)^{\prime}$ the  martingale difference sequence vector with variance-covariance matrix given by 
\begin{align}
\boldsymbol{\Sigma}  := \text{Var}( u_t ) 
= 
\begin{bmatrix}
\sigma^2_u  & \ \sigma_{uv}
\\
\sigma_{vu}  & \ \sigma_{vv}
\end{bmatrix}
\end{align}  
where $\boldsymbol{\Sigma}$ is a known positive-definite matrix.

\newpage 

Specifically, we are motivated in establishing the asymptotic validity of the bootstrap for the IVX estimator since its is employed in statistical problems for which the limiting distribution is non-standard such as when testing for structural breaks in predictive regression models and the bootstrap is necessary for deriving critical values. While for instance, we focus in the case of IVX-based predictability tests, the limit results we derive in this paper can be generalized for the supremum IVX-Wald statistic when testing for parameter instability in predictive regression models. The first Bootstrap algorithm we consider is described below. 

\medskip

\textbf{Bootstrap Algorithm}

\begin{itemize}
\item[\textbf{Step 1.}] Estimation Step
\begin{itemize}
\item[1.1.] Fit the predictive regression to the sample data $\left( Y_t, X_{t-1} \right)^{\prime}$ to obtain the OLS residuals $\hat{u}_t$, for $t \in \left\{ 1,..., n \right\}$. Obtain the estimate of $\hat{\upbeta}_n$.

\item[1.2.] Fit by OLS the AR(1) model with the LUR coefficient, to the regressor $X_t$, to obtain the OLS residuals $\hat{v}_t$, for $t \in \left\{ 1,..., n \right\}$. Set $\hat{v}_t = 0$. Obtain $\hat{\varrho}$. 
\end{itemize}
\end{itemize}

In this section, we assume that the autoregressive coefficient is given by $\rho_n = \left( 1 - \frac{\kappa}{n} \right)$, which allows to examine the limit theory of the  estimate of the autoregressive coefficient for moderate deviations from unity. Note that the estimate of $\varrho_n$ is obtained by ordinary least squares estimation which is expressed as 
\begin{align}
\hat{ \varrho}^{ols}_n = \left( \sum_{t=1}^n X^{2}_{t-1} \right)^{-1} \left( \sum_{t=1}^n X_t X_{t-1} \right)
\end{align}

\begin{itemize}

\item[\textbf{Step 2.}] Generate the bootstrap sample $\left\{ Y_t^{*}, X_t^{*}, t = 1,...,n \right\}$.
\begin{itemize}
\item[2.1.] Generate the bootstrap innovations $w_t^{*} = \left( u_{t}^{*},  v_{t}^{*\prime} \right)^{\prime} := \left( e_t \otimes \hat{u}_t, e_t \otimes \hat{v}_t \right)$, where the symbol $\otimes$ denotes element-wise vector multiplication, such that $e_t$ for $t \in \left\{ 1,..., n \right\}$, is a scalar sequence independent from the data such that $e_t \overset{ \textit{i.i.d} }{ \sim } \mathcal{N}(0,1)$.   

\item[2.2.] Generate $X_t^{*}$ such that 
\begin{align}
X_{t_b}^{*} = \hat{\varrho} X_{t_b -1}^{*} + u_{t_b}^{*}, \ \ \ \ \ \text{for}  \ t_b \in \left\{ 1,...,n \right\},
\end{align}
with initial conditions $X_0^{*} = 0$ and $b \in \left\{ 1,...,B \right\}$. 

\item[2.3.] Generate the associated bootstrap IVX instrument $z_t^{*}$ as below
\begin{align}
Z_0^{*} = 0 \ \ \ \text{and} \ \ \ Z_t^{*} = \sum_{j=0}^{t_b-1} \left( 1 + \frac{\kappa_z}{n} \right)^j  \Delta X_{t_b-j}^{*} , \ \ \text{for} \ \ t \in \left\{ 1,...,n \right\},
\end{align}
where $\kappa_z$ is the coefficient matrix for the persistence of the bootstrap IVX instruments that contains the same values as the original IVX instruments. 

\item[2.4.] Generate $Y_{ t_b }^{*}$ such that $y_{ t_b }^{*} =    \hat{\upbeta} X_{ t_b -1 }^{*}  + u_{ t_b }^{*}$. 
Set $Y_1^{*} = Y_1$.

\item[2.5.] Repeat Steps 2.1 to 2.4 $B$ times (e.g., $B = 1000$).  
\end{itemize}
 
\end{itemize}

\newpage 

\begin{remark}
Notice that the asymptotic validity of the bootstrap IVX estimator is an important property of general interest with applications in various inference problems. For example, \cite{katsouris2022partial} considers a framework for structural break testing in predictive regression models with persistent regressors. In the particular framework the bootstrap validity of the IVX estimator is necessary to ensure the consistent estimation of critical values from the bootstrap limiting distribution of the corresponding structural break test statistics. Furthermore, \cite{georgiev2021extensions} study the implementation of both the residual wild bootstrap as well as the fixed regressor wild bootstrap for the IVX estimator which is useful when constructing predictability tests (e.g., t-tests or Wald type statistics) based on the predictive regression model. 
\end{remark} 
Then, for the estimation step one can consider estimating $\beta_n$ either via the classical least squares estimation, denoted with $\hat{\beta}_n^{LS}$, or via the IVX instrumentation methodology proposed by \cite{phillipsmagdal2009econometric}, denoted with $\hat{\beta}_n^{IVX}$. The least squares estimate of $\beta_n$ is  
\begin{align}
\hat{\beta}_n^{LS} = \left( \sum_{t=1}^n X^2_{t-1} \right)^{-1} \left( \sum_{t=1}^n X_{t-1} Y_t \right)
\end{align}
The corresponding IVX estimate of $\beta_n$ is
\begin{align}
\hat{\beta}_n^{IVX} =  \left( \sum_{t=1}^n Z_{t-1} X_{t-1} \right)^{-1} \left( \sum_{t=1}^n Z_{t-1} Y_t  \right)
\end{align}
where 
\begin{align}
Z_{t} = \sum_{j=0}^{t-1} \left( 1 - \frac{\kappa_z}{n^{\gamma_z}}  \right)^j \left( X_{t-j} - X_{t-j-1} \right), \ \ \ \ \text{with} \ \kappa_z \in (0, + \infty), \gamma_z \in (0,1).
\end{align}  
We use for example, $\kappa_z = 1$ and $\gamma_z = 0.95$. 
\ 

From Step1A and Step1B, we can obtain the residual sequences, $ \left\{ \hat{u}_t^{IVX}, \hat{v}_t^{LS} \right\}_{t=1}^n$. To simplify the notation, for the remaining of this Section, we simply denote with $\hat{u}_t$ the corresponding residual sequence obtained via the IVX instrumentation procedure and $\hat{v}_t$ the residual sequence obtained via Least squares estimation. That is, 
\begin{align}
\hat{u}_t = \left( Y_t - \hat{ \beta}_n^{IVX} X_{t-1} \right)
\end{align}

Next, we define the centered residuals for both sequences, such that 
\begin{align}
\tilde{u}_t = \hat{u}_t- \frac{1}{n} \sum_{t=1}^n \hat{u}_t 
\\
\tilde{v}_t = \hat{v}_t- \frac{1}{n} \sum_{t=1}^n \hat{v}_t 
\end{align}   
Denote with $\tilde{F}_n$ the empirical distribution function based on $\left\{ \tilde{e}_t \right\}_{t=1}^n$ where $\tilde{e}_t = \left( \tilde{u}_t, \tilde{v}_t \right)^{\prime}$. Thus, $\tilde{F}_n$ associates mass $n^{-1}$ to each of $\tilde{e}_t, t =1,....,n$. Now, assuming that $\tilde{F}_n$ is the true distribution, draw a random sample $\left\{ e^{*}_t, t = 1,...,n  \right\}$ from $\tilde{F}_n$. Therefore, conditionally on $\left( X_1,..., X_{n-1}  \right)$, the random variables $\left\{ e^{*}_t, t = 1,...,n  \right\}$ are i.i.d with distribution function $\tilde{F}_n$.

\newpage

\underline{Step 2:} (Bootstrap Step)

We can now construct the bootstrap sample $\left\{ X_t^{*}, t = 1,...,n  \right\}$ recursively via the following expression 
\begin{align}
X_t^{*} = \hat{\rho}_n X_{t-1}^{*} + v^{*}_t, \ \ \ \ t = 1,...,n, 
\end{align} 
with $X_0 = 0$, where $\hat{\rho}_n$ it has been estimated in Step 1 by Least squares estimation.  

Next, we construct the bootstrap sample $\left\{ Y_t^{*}, t = 1,...,n  \right\}$ recursively via the following expression
\begin{align}
Y_t^{*} = \hat{\beta}_n X_{t-1}^{*} + u^{*}_t, \ \ \ \ t = 1,...,n, 
\end{align}
where $\hat{\beta}_n$ it has been estimated in Step 1 by the IVX instrumentation methodology. 

Next, we re-estimate the model parameter of interest, denoted with  $\hat{\hat{\beta}}_n$, by applying again the IVX instrumentation on the generated sequence of $\left\{ Y_t^{*}, X_t^{*} \right\}$, from Step 1. That is,   
\begin{align}
\hat{ \hat{\beta}}_n^{*} =\left( \sum_{t=1}^n Z^{*}_{t-1} X^{*}_{t-1} \right)^{-1} \left( \sum_{t=1}^n Z^{*}_{t-1} Y^{*}_t  \right) 
\end{align}
where 
\begin{align}
Z^{*}_{t} = \sum_{j=0}^{t-1} \left( 1 - \frac{\kappa_z}{n^{\gamma_z}}  \right)^j \left( X^{*}_{t-j} - X^{*}_{t-j-1} \right), \ \ \ \ \text{with} \ \kappa_z \in (0, + \infty), \gamma_z \in (0,1).
\end{align}  
We use for example, the same values as in Step 1, that is, $\kappa_z = 1$ and $\gamma_z = 0.95$.
\

\medskip

\underline{Step 3:} (and onwards)
\

For the remaining bootstrap steps, that is, $b = \left\{ 1,...,B \right\}$ we generate a sequence of IVX estimates that is $\underline{\beta}^{*IVX} = \left( \beta^{*IVX}_{n1}, ...., \beta^{*IVX}_{nB} \right)$.
 
\medskip

The main aim of this Section, is to derive the limit distribution of $\underline{\beta}^{*IVX}$ given $\left( X_1, ... X_n \right)$ and to show that it is the same as the original limit distribution of the IVX estimator. Notice that in practise the bootstrap procedure should be able to capture the time series dependence structure. It is worth mentioning that we need to employ some autoregressive robust covariance matrix estimation methodology when applying our proposed bootstrap algorithm to ensure that the dependence structure of the nonstationary time series model is preserved which safeguards the consistent estimation of the Wald-type statistic (\cite{jansson2004error} and \cite{jansson2012nearly}).

\newpage

\subsection{Asymptotic distribution result.}

We consider a triangular array of paired random variables such that $\left\{ X_{k,n}, Y_{k,n} \right\}$ with $n \geq 1$ and $k \geq 1$ satisfying
\begin{align}
\label{system1}
Y_{k,n} &= \upbeta_n X_{k-1,n} + u_{k}, 
\\
\label{system2}
X_{k,n} &= \varrho_n X_{k-1,n} + v_{k}, \ \ \text{with} \ \ X_0 = 0.
\end{align}
where  $e_k = \left( u_k, v_k \right)^{\prime} \sim \ \mathcal{N} \big( 0, \boldsymbol{\Sigma} \big)$. Then, we construct the following triangular array using endogenous information generated by the triangular array $\left\{ X_{k,n} \right\}_{n, k \geq 1}$ given by 
\begin{align}
\label{instrument}
Z_{k,n} = \sum_{j=0}^{t-1} \left( 1 - \frac{\kappa_z}{n^{\upgamma_z}}  \right)^j \big( X_{k-j, n} - X_{k-j-1,n} \big), \ \ \ \ \text{with} \ \kappa_z \in (0, + \infty), \upgamma_z \in (0,1).
\end{align} 
We define the following random sequence
\begin{align}
\mathcal{J}_{n} :=  \left( \sum_{t=1}^n Z_{t-1} X_{t-1} \right)^2 \left( \hat{\upbeta}_n^{ivx} - \upbeta_n \right) 
\end{align} 
where $\hat{\upbeta}_n^{ivx}$ is the IVX estimator of $\upbeta_n$. Therefore, the random measure $\mathcal{J}_{n}$ can be shown to have a limiting distribution which is given by the following expression 
\begin{align}
\label{to.show}
\mathcal{J}_{n} \overset{d}{\to} \mathcal{J} :=  - \kappa_z \ \mathcal{U} \left( . \right) \left\{ \omega_{xx}^2 +  \int_0^1 J_{\kappa}(s) dJ_{\kappa}(s) \right\} 
\end{align}
where $\mathcal{U} \left( . \right)$ is Brownian motion with variance $\sigma^2_{uu} \tilde{V}$ with $\tilde{V} \equiv - \omega^2_{xx} \big/ 2 \kappa_z$, for $\upgamma_x > \upgamma_z $ since $\upgamma_z \in (0,1)$. Then, the bootstrap estimator $\hat{\upbeta}_n^{*ivx}$ of $\upbeta_n$ is obtained from the corresponding bootstrap sequence $\left\{ Y_t^{*}, X_t^{*} \right\}$ which implies the following random sequence
\begin{align}
\mathcal{J}_{n}^{*} =  \left( \sum_{t=1}^n Z^{*}_{t-1} X^{*}_{t-1} \right)^2 \left( \hat{\upbeta}_n^{*ivx} - \upbeta_n \right) 
\end{align} 
where 
\begin{align}
Z^{*}_{t} = \sum_{j=0}^{t-1} \left( 1 - \frac{\kappa_z}{n^{\upgamma_z}}  \right)^j \left( X^{*}_{t-j} - X^{*}_{t-j-1} \right), \ \ \ \ \text{with} \ \kappa_z \in (0, + \infty), \upgamma_z \in (0,1).
\end{align}  
Therefore, in practise we need to show that $\mathcal{J}_{n}$ and $\mathcal{J}_{n}^{*}$ have the same limit distribution (we define this step as Condition 1). To do this, we consider the random measure $\mathcal{\zeta}_n$ that corresponds to the paired triangular array \eqref{system1}-\eqref{system2} and the generated instrument is given in expression \eqref{instrument}. \cite{sweeting1989conditional} proposed a framework for evaluating conditional weak convergence. Specifically, when considering the asymptotic validity of the bootstrap algorithm it is necessary to consider the conditional weak convergence based on the bootstrapped invariance principles. On the other hand, the dependence structure of the predictive regression model remains fixed therefore the invariance principles based for the probability space that corresponds to the bootstrapped observations is valid due to similar conditions.

\newpage

Notice also that although the practitioner needs to choose the nuisance parameter of persistence for the instrumental variable. In other words, the statistics of interest such as the model estimator as well as the Wald test are calculated based on the following expression:  
\begin{align}
\mathcal{\zeta}_n :=  \left( \sum_{k=1}^n Z_{k-1,n} X_{k-1,n} \right)^2 \left\{ \frac{ \displaystyle \sum_{k=1}^n Z_{k-1,n} Y_{k,n}  }{ \displaystyle  \sum_{k=1}^n Z_{k-1,n} X_{k-1,n} } - \upbeta_n  \right\}
\end{align}
Following a similar approach in the literature (to add references), we need to specify an auxiliary functional of the form $\Psi(.)$ such that 
\begin{align}
\label{condition}
\mathcal{H} ( . , x ) = \mathbb{P} \big( \mathcal{J}_{n} \leq x \big), 
\end{align}
so that we can derive an associated result that shows a probability limit of the form 
\begin{align}
\underset{ n \to \infty }{ \mathsf{lim} } \mathbb{P}_{ \upbeta_n } \big( \mathcal{\zeta}_n \leq x \big) = \mathcal{H} ( . , x )  
\end{align}
To do that, we can define for example with 
\begin{align}
\label{condition22}
\mathcal{H}_n \left(  \hat{\upbeta}^{ivx}_n, x \right) = \mathbb{P} \bigg( \mathcal{J}_{n}^{*} \leq x   \big| X_1,..., X_n \bigg)
\end{align} 
which is taken to be a regular conditional probability distribution function.

\begin{remark}
The above statements imply that to show that the bootstrap approximation is valid, then we shall show that along almost all paths of $H_n$ given by expression \eqref{condition22} converges in distribution to the distribution of $\mathcal{J}$ given by expression \eqref{to.show} (we define this step as Condition 2). Therefore, to prove the result we are aiming we can try to prove that a corresponding weakly convergence result holds for the quantity $\mathcal{J}_n^{*}$.
\end{remark}

Notice that in the literature for bootstrap based inference of time series models (see, \cite{basawa1991bootstrapping}) the authors prove that the bootstrap estimator of the autoregressive coefficient in an AR(1) model in the explosive case will not have the same limiting distribution as the original estimator, within our framework this might not be affecting the result which we are aiming to prove. In other words, due to the structure of the predictive regression system, in each bootstrap step, to generate the pair $\left\{ X_t^{*}, Y_t^{*} \right\}$ we use the residuals $u_t^{*}$ obtained from the first step regression and then  construct the observations $Y_t^{*}$ in the second step regression. Therefore, in practise since we are not interested whether or not the bootstrap estimator of the autoregressive coefficient (under various nonstationary processes) is valid or not, we can still use the particular residual sequences. Therefore, after fitting the predictive regression we can obtain the bootstrap IVX estimator conditional on the nonstationary process. More precisely, the statistical validity of the bootstrap residuals that correspond to the autoregressive process is proved to hold by \cite{paparoditis2003residual} (see, Appendix \ref{FCLT.bootstrap.Appendix}). 

\newpage 

Therefore, we can consider the covariance matrix of the random variable $\psi^{*}$, where 
\begin{align}
\psi_n^{*} := n^{ ( 1 + \upgamma_z ) / 2 } \left( \hat{\upbeta}_n^{*IVX} - \upbeta_n \right) \Rightarrow \psi^{*}  := \mathcal{MN} \left( 0, \tilde{ \boldsymbol{\Sigma} }^{*} \right)
\end{align}

Therefore, we consider the martingale difference sequence $\xi^{*}_{nt} :=  \left( \frac{1}{ n^{ ( 1 + \upgamma_z ) / 2 } } Z^{*}_{t-1} u^{*}_t , \frac{1}{ \sqrt{n} } v^{*}_t  \right)^{ \prime }$. 
\

Then the martingale conditional variance has the following structure 
\begin{align}
\label{mds.covariance}
\sum_{t=1}^n \mathbb{E} \big( \xi^{*}_{nt} \xi_{nt}^{*\prime} | \mathcal{F}_{nt-1} \big) =
\begin{bmatrix}
\displaystyle \left( \frac{1}{ n^{ 1 + \gamma_z } } \sum_{t=1}^n Z^{*2}_{t-1} \right) \sigma^{*2}_{uu} \ \ & \ \ \displaystyle \left( \frac{1}{ n^{1 + \frac{\gamma_z}{2}}} \sum_{t=1}^n Z^{*}_{t-1} \right) \sigma^{*}_{uv} 
\\
\\
\displaystyle \left( \frac{1}{ n^{1 + \frac{\gamma_z}{2}}} \sum_{t=1}^n Z^{*}_{t-1} \right) \sigma^{*}_{vu} \ \  & \ \ \displaystyle \sigma^{*2}_{vv}
\end{bmatrix}
\end{align}
\begin{remark}
In this paper, in contrary to some related studies in the literature such as \cite{georgiev2018testing}, \cite{cavaliere2020inference} and \cite{georgiev2021extensions} we do not consider the implementation of the fixed regressor bootstrap. The reason for this is simple, since we study estimators for the predictive regression model based on the paired data $\left\{ Y_t, X_t \right\}_{t=1}^n$, then keeping the regressors fixed in each bootstrap iteration asymptotically might not preserve the limiting distribution of the original pair.   
\end{remark}


\section{Monte Carlo Simulations.}
\label{Section4}

\begin{example}
We consider the following simple data-generating process (DGP):
\begin{align}
y_t &= \beta_0 x_t + u_t, \ \ \ t = 1, 2,..., n
\\
x_t &= \left( 1 - \frac{c}{n} \right) x_{t-1} + v_t, \ \ \ \text{for some} \ c > 0.
\end{align}
where $v_t \sim \mathcal{N} \left( 0, \sigma_v^2 \right)$ and $u_t = \rho u_{t-1} + \epsilon_t$ where $\epsilon_t \sim \mathcal{N} \left( 0, \sigma_{\epsilon}^2 \right)$. In practise, when $\rho = 0$ then $u_t = \epsilon_t$ is the \textit{i.i.d} residual process. 
\end{example}
The particular example can demonstrate our asymptotic results presented in this paper. Then the performance of our bootstrap procedure can be assessed using total variation distances applied to the empirical distribution function and the corresponding empirical distribution function of the residual process (see related limit theory in \cite{stroud1972fixed} and 
\cite{mykland1992asymptotic}). According to \cite{horowitz2003bootstrap} bootstrap resampling is a method for estimating the distribution of an estimator or test statistic by resampling one's data or a model estimated from the data. Under conditions that hold in a variety of econometric applications, the bootstrap provides approximations to distributions of statistics, coverage probabilities of confidence intervals, and rejection probabilities of tests that are more accurate than the approximations of first-order asymptotic distribution theory.



\newpage

\section{Conclusion.}
\label{Section5}

In this paper we contribute to the literature of bootstrap methodologies for nonstationary time series models. Specifically, we study the bootstrapping problem of nonstationary autoregressive processes with predictive regression models with near unit root (persistent) regressors. These are regressors which are near-integrated, that is, close to the unit boundary. We establish the asymptotic validity of the bootstrap based estimators and investigate the finite-sample behaviour of these bootstrap procedures by Monte Carlo experiments. The related theory to the bootstrap aspects of the paper is presented in the \hyperref[appn]{Appendix}. Moreover, the IVX instruments are constructed based on endogenous information therefore satisfying the relevance condition for valid instruments. Therefore, the role of the bootstrap is to mimic the dependence structure in the predictive regression model. A discussion regarding the relevance condition of instrumental variables can be found in the paper of \cite{hall1996judging}. 

Therefore, one of the applications of our methodology is to employ the bootstrap for nonstandard and nonpivotal testing problems. Generally a Wald-type statistic is a pivotal function, as occurs for instance when testing linear restrictions on the coefficients of linear regressions. On the other hand, when testing in the predictive regression model we employ the IVX-Wald test which is robust to the nuisance parameter of persistence. We have investigated two estimation methodologies when bootstrapping nonstationary autoregressive processes with predictive regression models for the purpose of obtaining the bootstrapped distribution of test statistics. On the other hand, bootstraps with exchangeable weights are not so wide spread in the econometrics and statistics literature. Thus considering the random weight bootstrap in the context of nonstationary autoregressive processes is a novel contribution to the literature. 

As future research we aim to investigate the bootstrap validity of the framework that corresponds to the uniform inference methodology in predictive regression. In particular, the quasi restricted likelihood ratio test (QRLRT) shows to have desirabe properties. More precisely, the restricted likelihood has been found to provide a well-behaved likelihood ratio test in the predictive regression model even when the regressor exhibits almost unit root behaviour. Another interesting further research worth mentioning include to consider the bootstrap validity of additional persistence classes or different functional forms for the predictive regression model as proposed in the framework of \cite{duffy2021estimation}. Our proposed bootstrap algorithm can be also employed in different settings as well such that within a sequential monitoring framework for structural breaks in Garch $(p,q)$ models (see, \cite{berkes2004sequential}). Furthermore, although in this paper we consider the construction of Wald-type statistics based on linear restrictions other transformations can be also considered such as non-linear restrictions which are found to have better power performance (see also \cite{heimann1996bootstrapping}).

\begin{small}
\paragraph{Acknowledgements}

I wish to thank Jose Olmo, Tassos Magdalinos and Jean-Yves Pitarakis for helpful discussions during my PhD studies as well as Stathis Paparoditis and Giuseppe Cavaliere. The author of this article acknowledge the use of the IRIDIS High Performance Computing Facility and associated support services at the University of Southampton, in the completion of this work. Financial support from the VC PhD studentship of the University of Southampton is also gratefully acknowledged.
\end{small}

\newpage

\begin{appendix}

\begin{center}
\textbf{APPENDIX}
\end{center}

\section{Technical Proofs}
\label{Section6}

(page S.23 in \cite{georgiev2021extensions})

We have that 
\begin{align}
\frac{1}{ n^{1 + \upgamma_z} } \sum_{t=1}^{ \floor{nr} } Z_{t-1} X_{t-1} = \frac{\omega}{a} \frac{1}{n} \sum_{t=2}^{ \floor{nr} } v_{t-1} \xi_{t-1} - \frac{c}{a} \frac{1}{n^2} \sum_{t=2}^{ \floor{nr} } \xi_{t-2} \xi_{t-1} + o_p(1),  
\end{align}

It then follows by near-integration asymptotics that 
\begin{align*}
\frac{1}{ n^{1 + \upgamma_z} } \sum_{t=1}^{ \floor{nr} } Z_{t-1} X_{t-1} 
&\Rightarrow 
\frac{ \omega^2 }{ a } \left( \int_0^{\pi} J dM_c + \left[ M_v \right]_{\pi} - c \int_0^{\pi} J^2_c (r) \right)
\\
&=
\frac{ \omega^2 }{ a }  \left( \int_0^{\pi} J dM_c + \left[ M_v \right]_{\pi} \right),
\end{align*}

Notice that it holds that $\left[ M_v \right]_{\pi} = \left[ J_c \right]_{\pi}$, that is, the two stochastic processes have the same quadratic variation. Furthermore, it holds that 
\begin{align}
J^2_{c} ( \pi ) - \int_0^{\pi} J_c (r) dJ_c (r) = \int_0^{\pi} J_c (r) dJ_c (r) + \left[  J_c \right]_{\pi}
\end{align}
which holds by the semi-martingale property of $J_c(r)$. Moreover, notice these stochastic processes are monotonically increasing with a continuous quadratic variation function $\left[ M_v \right]( \pi )$ on a compact set $\mathcal{B} (0,1)$, then it is sufficient to show that the asserted weakly convergence of the sample moments to the corresponding stochastic integrals holds point-wise in probability for $\pi \in [0,1]$.

\medskip

Consider the proof from \cite{li2001bootstrapping}. In particular, we are interested to show the invariance principles for the partial sums of $\varepsilon_t^{*}$ and $u_t^{*}$. Thus, first we show that 
\begin{align}
\frac{1}{ \sqrt{T} } \sum_{t=1}^{ \floor{nr} } \varepsilon_t^{*} \Rightarrow B_{ \varepsilon } (r)
\end{align}  
Let EDF denote the empirical distribution function $d_2(.,.)$ be the Mallows metric defined as $d_p( \nu, \mu) = \mathsf{inf} \ \mathbb{E} \left( \norm{ V - U }^p \right)^{ 1 / p}$. Thus, we denote with 
\begin{align}
S_{ \floor{nr} } := \frac{1}{ \sqrt{T} } \sum_{t=1}^{ \floor{nr} } \varepsilon_t^{*} \Rightarrow B_{ \varepsilon } (r), \ \ S_{ \floor{nr} }^{*} := \frac{1}{ \sqrt{T} } \sum_{t=1}^{ \floor{nr} } \varepsilon_t^{*} \Rightarrow B_{ \varepsilon } (r) 
\end{align} 
Thus, we need to show that 
\begin{align}
d_2 ( F, \hat{F}_T ) \to 0
\end{align}


Moreover, we need to establish convergence of finite-dimensional distributions and verify the tightness condition. In particular, for the finite-dimensional convergence, it is sufficient to show that for any finite $k$ and $0 < r_1 < ... < r_k < 1$, if
\begin{align}
V_k = \big[ S_{ \floor{nr_1} }, S_{ \floor{nr_2} } - S_{ \floor{nr_1} } , ...., S_{ \floor{nr_k} } - S_{ \floor{nr_{k-1}} } \big]
\end{align}
and 
\begin{align}
V_k^{*} = \big[ S^{*}_{ \floor{nr_1} }, S^{*}_{ \floor{nr_2} } - S^{*}_{ \floor{nr_1} } , ...., S^{*}_{ \floor{nr_k} } - S^{*}_{ \floor{nr_{k-1}} } \big]
\end{align}
then $d_2 \left( V_k^{*}, V_k \right) \to 0$. Thus, it can be shown that $d_2 \left( V_k^{*}, V_k \right)^2 \leq d_2 \left( F_T, \tilde{F}_T \right)^2 \to 0$.  

For tightness, it suffices to show that there exists a non-decreasing function $\psi$ such that, for almost all sample paths, for $0 \leq r_1 \leq r \leq r_2 \leq 1$, 
\begin{align}
\mathbb{P}^{*} \left( \left| S^{*}_{\floor{nr}} - S^{*}_{\floor{nr_1}} \right| \right)
\end{align}  

\medskip

From the paper: Testing for parameter instability in predictive regression models

\begin{remark}
Notice that considering for example the test statistic $\mathcal{W}^{*}_{ivx}$ we observe that bootstrap invariance principles hold jointly along with stated convergence results. In practise, in terms of convergence in topological spaces, we have joint weak convergence of random measures. The concept is weaker than weak convergence in probability, although it reduces to the latter when the limit distribution is non-random.  
\end{remark}


\section{Smoothing Local-to-moderate unit root theory.}

Consider an autoregressive process with local-to-moderate deviations from unit root of the form $\rho_n = \left( 1 + \frac{c}{K} \right)$, where $K$ converges to infinity with the sample size $n$ and $\rho_n$ approaches unity from the stationary or the explosive side according to the sign of $c$ (see, \cite{Phillips2010smoothing}). 

We consider that such a time series process constitutes of $m$ blocks of $K$ observations with total sample size $n = mK$. Then, partitioning the chronological sequence $\left\{ t = 1,...,n \right\}$ by setting $t = [kj] + k $ for $k \in \left\{ 1,...,K \right\}$ and $j \in \left\{ 0,..., m-1 \right\}$, it is possible to study the asymptotic behaviour of the time series $\left\{ X_t : t = 1,...,n  \right\}$ via the asymptotic properties of the time series $\left\{  X_{ Kj + k} : j = 0,..., m-1, k = 1,..., K \right\}$. The latter representation is particularly useful for revealing the transition from non-stationary to stationary autoregression as the number of blocks $m$ increases. 

More formally, the process with the above characteristics may be written in the form 
\begin{align}
X_t &= \rho_{n,m} X_{t-1} + u_t, \ \ \ u_t \sim _{iid} \left( 0, \sigma^2  \right), \\
\rho_{n,m}  &= \left( 1 + \frac{c}{K} \right) = \left( 1 + \frac{cm}{n} \right)
\end{align}
since we have $m$ blocks with $K$ observations in each block, giving us a total of $n = mK$.

Therefore, under this setting when $m=1$, then the usual local to unity model applies, and when $m \to \infty$, then the moderate deviation theory of PM and GP holds.   

Let $W$ be a standard Brownian motion and $J_c(t) = \displaystyle \int_{0}^t e^{c(t-s)} dW(s)$ be a corresponding OU process. For each $m$, letting 
\begin{align}
\widetilde{W} (t) = \sqrt{m} W \left( \frac{t}{m} \right),
\end{align}   
we observe that W also a standard Brownian motion and we denote by $\widetilde{J}_c(t) = \int_0^t e^{c(t-s)} d \widetilde{W} (s)$ the associated OU process. Thus, for given $m \geq 1$, we may derive a limit theory for the least squares estimate $\hat{\rho}_{n,m}$ of $\rho_{n,m}$ as $n \to \infty$ using earlier results from standard local to unity asympotics.


Thus, using the identities above we obtain 
\begin{align}
\int_0^1 J_{cm} (s) dW(s) 
&= \frac{1}{m} \int_0^m \widetilde{J}_c (s) d \widetilde{W} (s),
\\
\int_0^1 J_{cm} (s)^2 ds 
&= \frac{1}{m^2} \int_0^m \widetilde{J}_c (s)^2  ds.
\end{align} 
the results derived in Phillips  imply that, for fixed $m$ and $n \to \infty$, the asymptotic distribution of the least squares estimator takes the following form 
\begin{align}
n \left( \hat{\rho}_{n,m} - \rho_{n,m} \right) \Rightarrow \frac{ \displaystyle \int_0^1 J_{cm} (s) dW(s) }{ \displaystyle \int_0^1 J_{cm} (s)^2 ds } = m \frac{ \displaystyle \int_0^m \widetilde{J}_c (s) d \widetilde{W}(s)}{ \displaystyle \int_0^m \widetilde{J}_c (s)^2 ds}. 
\end{align}
When $c < 0$, sequential limits with $n \to \infty$ followed by $m \to \infty$ lead to the normal asymptotic theory given in PM and GP:

\begin{align}
\frac{n}{\sqrt{m}} \left( \hat{\rho}_{n,m} - \rho_{n,m} \right) 
&\Rightarrow \frac{ \displaystyle \frac{1}{\sqrt{m}} \int_0^m \widetilde{J}_c (s) d \widetilde{W}(s)}{ \displaystyle \frac{1}{m} \int_0^m \widetilde{J}_c (s)^2 ds} \ \ \text{for fixed} \ m
\\
&= \frac{ \displaystyle \frac{1}{\sqrt{m}}  \sum_{j=1}^m \int_{j-1}^j \widetilde{J}_c (s) d \widetilde{W}(s)}{ \displaystyle \frac{1}{m}  \sum_{j=1}^m \int_{j-1}^j \widetilde{J}_c (s)^2 ds}
\\
& \Rightarrow \frac{ \mathcal{N} \left( 0, - \frac{1}{2c} \right) }{ - \frac{1}{2c} } \equiv \mathcal{N} ( 0, -2c ) \ \ \text{as} \ \  m \to \infty.
\end{align}


\section{Bootstrapping Unstable Autoregressive Processes.}\label{appA}

\subsection{Theoretical Background.}

In this section, we present related asymptotic theory examples as preliminary theory to the framework proposed in this paper. 

Consider the first-order autoregressive process $\left\{ X_t \right\}$, for $t = 1,2,...$,  
\begin{align}
X_t = \beta X_{t-1} + \epsilon_t, \ \ \ X_0 = 0,
\end{align} 
where $\left\{  \epsilon_t \right\}$ are independent $\mathcal{N} (0,1)$ random variables. 

The least squares estimator $\hat{\beta}_n$ of $\beta$, based on a sample of $n$ observations $\left( X_1,..., X_n \right)$, is given by  the following expression 
\begin{align}
\label{estimator}
\hat{\beta}_n = \left( \sum_{t=1}^n X^2_{t-1} \right)^{-1} \left( \sum_{t=1}^n X_t X_{t-1} \right)  
\end{align}

Note that the asymptotic validity of the bootstrap estimator corresponding to $\hat{\beta}_n$ for the stationary case, viz., $| \beta| < 1$, follows from the work of \cite{bose1988edgeworth} and the validity for the explosive case, viz., $|\beta | > 1$ has recently established by \cite{basawa1989bootstrapping}. Both of these papers consider the general case when the distribution of $\left\{ \epsilon_t \right\}$ is not necessarily known. The limit distribution of $\hat{\beta}_n$ in the unstable case $| \beta | = 1$ is known to be nonnormal. It is therefore, of special interest to consider the bootstrap approximation for the distribution of $\hat{\beta}_n$ for the unstable case. 

\subsection{Invalidity of the bootstrap estimator.} 
\

Let $\xi_n = \left( \sum_{t=1}^n X^2_{t-1} \right)^{1/2} \left( \hat{\beta}_n - \beta \right)$, where $\hat{\beta}_n$ is defined as above. It is then, well known that when $\beta = 1$, 
\begin{align}
\label{distribution}
\xi_n \overset{ d }{ \to } \xi = \frac{1}{2} \left[ W^2(1) - 1 \right] \left\{ \int_0^1 W^2 (t) dt \right\}^{-1 / 2} \ \ \text{as} \ \ n \to \infty, 
\end{align}
where $W(t)$ is a standard Wiener process. Then, the corresponding bootstrap sample $\left\{ X_t^{*} \right\}$ is obtained recursively from the following relation 
\begin{align}
X_t^{*} = \hat{\beta}_n X_{t-1}^{*} + \epsilon_t^{*}, \ \ X_0^{*} = 0. 
\end{align}
where $\left\{ \epsilon_t^{*} \right\}$ constitutes a random sample from $\mathcal{N} \left( 0, 1 \right)$. Then, the bootstrap estimator $\hat{\beta}^{*}_n$ of $\beta$ is then defined as in \eqref{estimator} with $X$'s replaced by $X^{*}$'s.


Let $\xi_n^{*} =  \left( \sum_{t=1}^n X^{*2}_{t-1} \right)^{1/2} \left( \hat{\beta}^{*}_n - \beta \right)$ denote the bootstrap version of $\xi_n$. It will be shown that $\xi_n$ and  $\xi_n^{*}$ do not have the same limit distribution, thus invalidating the bootstrap. 

Therefore, we consider a triangular array $\left\{ X_{k,n}, k \leq 1, n \leq 1 \right\}$ satisfying
\begin{align}
\label{model2}
X_{k,n} = b_n X_{k-1, n} + \epsilon_k, \ \ X_0 = 0,
\end{align}   
with independent $\epsilon_k \sim \mathcal{N} \left( 0, 1 \right)$ and where $\left\{ b_n \right\}$ is a sequence of numbers such that $n \left( b_n - 1 \right) \to \gamma$. Let
\begin{align}
\Psi ( \gamma ) = \frac{ \displaystyle \int_0^1 \left( 1 - t + te^{-2 \gamma} \right)^{-1} W(t) dW(t) }{ \displaystyle \left\{ \int_0^1 \left( 1 - t + te^{-2 \gamma} \right)^{-2} W^2(t) dt \right\}^{ 1/2 } }
\end{align}
where $\left\{ W(t) : 0 \leq t \leq 1 \right\}$ is a standard Brownian motion and 
\begin{align}
\label{condition}
H ( \gamma, x ) = \mathbb{P} \left( \Psi ( \gamma ) \leq x \right). 
\end{align}
Then, by Theorem 1 of \cite{chan1987asymptotic1} we have that
\begin{align}
\underset{ n \to \infty }{ \text{lim} } \mathbb{P}_{b_n} \left( \zeta_n \leq x \right) = H ( \gamma, x )  
\end{align}
where 
\begin{align}
\zeta_n = \left( \sum_{k=1}^n X^2_{k-1,n} \right)^{1/2} \left( \frac{ \sum_{k=1}^n X_{k,n}  X_{k-1,n}  }{ \sum_{k=1}^n X^2_{k-1,n}  } - b_n \right)
\end{align}
and where $\mathbb{P}_{b_n}$ indicates the probability distribution induced by the model in \eqref{model2}. Define 
\begin{align}
\label{condition2}
H_n \left(  \hat{b}_n, x \right) = \mathbb{P} \bigg( \xi_n^{*} \leq x   \big| X_1,..., X_n \bigg)
\end{align} 
which is taken to be a regular conditional probability distribution function. Therefore, we define a random measure
\begin{align}
\eta_n ( A ) = \int_A H_n \left(  \hat{b}_n, dx \right). 
\end{align}
Since $H( \gamma, x )$ given by \eqref{condition} is continuous in $\gamma$ for each fixed $x$, we have
\begin{align}
\eta ( A ) = \int_A H \left( \xi^{\prime}, dx \right)
\end{align}
\begin{align}
\xi^{\prime} = \frac{1}{2} \left[ W^2(1) - 1 \right] \left\{ \int_0^1 W^2 (t) dt \right\}^{ -1 } \ \ \text{as} \ \ n \to \infty, 
\end{align}


Therefore, if the bootstrap approximation were valid then along almost all paths $H_n$ given by \eqref{condition2} would converge in distribution to the distribution of $\xi$ given in \eqref{distribution}. However, we have in fact that  
\begin{align}
\label{convergence}
\eta_n \Rightarrow \eta, \ \ \text{as} \ \ n \to \infty.
\end{align}

\begin{proof}
By almost sure representativeness of convergent laws, it is possible to define $\tilde{\beta}_n, n \leq 1$, and $\tilde{\xi}$ with $\hat{\beta}_n = \tilde{\beta}_n$, $\xi^{\prime} =_{d} \tilde{\xi}$ and
\begin{align}
n \left( \tilde{\beta}_n - 1 \right) \overset{ d }{ \to } \tilde{ \xi } \ \ \textit{almost surely} \ \text{as} \ n  \to \infty.
\end{align}
Recall that, 
\begin{align}
n \left( \hat{\beta}_n - 1 \right) \overset{ d }{ \to } \xi^{\prime} \ \text{as} \ n  \to \infty  
\end{align}
Therefore, we have that 
\begin{align}
\label{condition3}
H_n \left(  \tilde{ \beta }_n, x \right) \to H ( \tilde{ \xi }, x ) \ \ \textit{almost surely} \ \text{as} \ n  \to \infty. 
\end{align}

Hence for sets of the form 
\begin{align}
A_j = \bigcup_{i=1}^{ m_j } \big( x_{ij}, y_{ij} \big)
\end{align}
representing a disjoint union of intervals, \eqref{condition3} implies that 
\begin{align}
\label{convergence1}
\bigg( H_n \left( \tilde{ \beta }_n, A_1 \right),..., H_n \left(  \tilde{ \beta }_n, A_k \right) \bigg) \to \bigg( H_n \left( \tilde{ \xi }, A_1 \right),..., H_n \left( \tilde{ \xi }, A_k \right) \bigg)
\end{align}
\textit{almost surely} as $n \to \infty$.

Since, 
\begin{align}
\bigg( \eta_n ( A_1 ), .... , \eta_n ( A_k )  \bigg) =_{d} \bigg( H_n \left( \tilde{ \beta }_n, A_1 \right),..., H_n \left(  \tilde{ \beta }_n, A_k \right) \bigg)
\end{align}
and
\begin{align}
\bigg( \eta ( A_1 ), .... , \eta ( A_k )  \bigg) =_{d} \bigg( H_n \left( \tilde{ \xi }, A_1 \right),..., H_n \left( \tilde{ \xi }, A_k \right) \bigg),
\end{align}
which implies that \eqref{convergence} follows from \eqref{convergence1}.
\end{proof}

The above example, demonstrates that the least squares estimate of the autoregressive coefficient in the case of the unstable autoregression model ($\beta = 1$), induces an invalidity of the bootstrap procedure. A similar invalidity of the bootstrap procedure occurs when considering the distribution of $n \left( \hat{\beta}_n - \beta \right)$ as examined by \cite{chan1987asymptotic1}. In this paper, we aim to demonstrate that in the case of predictive regression model, when we estimate the unknown parameter using the IVX instrumentation procedure, then the limiting distribution of the original estimate is preserved which verifies the validity of the corresponding IVX bootstrap estimate.     

\section{Bootstrapping Explosive Autoregressive Processes.}\label{appB}

\subsection{Preliminary results and the bootstrap estimate.}

First, we review some basic limit results for the explosive case as in  \cite{bose1988edgeworth}. We assume that $| \beta | > 1$ and define with 
\begin{align}
\label{defintion}
U_n = \sum_{t=1}^n \beta^{ - (t - 1 )} \epsilon_t \ \ \text{and} \ \ V_n = \sum_{t=1}^n \beta^{ - (n - t )} \epsilon_t
\end{align}
where $\left\{ \epsilon_t \right\}$ are i.i.d random variables with $\mathbf{E} \left( \epsilon_j  \right) = 0$ and Var$\left( \epsilon_j \right) = \sigma^2$. Moreover, it can be shown that $U_n$ and $V_n$ are identically distributed for each $n$ and there exist random variables $U$ and $V$ such that 
\begin{align}
\left( U_n, V_n \right) \to_{d} \left( U, V \right) \ \ \text{as} \ \ n \to \infty, 
\end{align}  
where $U$ and $V$ are independent and identically distributed. If $\left\{ \epsilon_t \right\}$ are normal, $U$ and $V$ are normal each with mean 0 and variance $\left( 1 - \beta^{-2}  \right)^{-1}$. In the general case, $U$ and $V$ can be represented by 
\begin{align}
U \overset{d}{=} \sum_{t=1}^{\infty} \beta^{-(t-1)} \epsilon_t \overset{d}{=} V,
\end{align}
Note that the common characteristic function of $U$ and $V$ is
\begin{align}
\phi( \tau ) = \prod_{t=1}^{ \infty } \phi_{ \epsilon } \left( \beta^{-(t-1)} \tau \right) 
\end{align}
where $\phi_{\epsilon}( . )$ is the characteristic function of $\epsilon_t$.

Let $\hat{ \beta}_n$ denote the least squares estimate of $\beta$. The following Theorem summarizes the limit distribution of $\hat{ \beta}_n$ for the case $| \beta | > 1$.
\begin{theorem}
\label{theorem.Anderson}
For $\hat{ \beta}_n$ defined in the autoregressive model, we have that for $| \beta | > 1$, 
\begin{align}
\left( \beta^2 - 1 \right)^{-1} | \beta |^n \left( \hat{ \beta}_n - \beta   \right) \to_{d} V \big/ U,
\end{align}
where $U$ and $V$ are defined in \eqref{defintion}.
\end{theorem} 


We now describe the corresponding bootstrap estimate. Let $\hat{\epsilon}_t = X_t - \hat{ \beta}_n X_{t-1}$ and define $\tilde{\epsilon}_t = \hat{\epsilon}_t- n^{-1} \sum_{t=1}^n \hat{\epsilon}_t$, the centered residuals. Denote by $\tilde{F}_n$ the empirical distribution function based on $\left\{ \tilde{\epsilon}_t, t = 1,...,n  \right\}$. Thus, $\tilde{F}_n$ associates mass $n^{-1}$ to each of $\tilde{\epsilon}_t, t =1,....,n$. Now, assuming that $\tilde{F}_n$ is the true distribution, draw a random sample $\left\{ \epsilon^{*}_t, t = 1,...,n  \right\}$ from $\tilde{F}_n$. Therefore, conditionally on $\left( X_1,..., X_n  \right)$, the random variables $\left\{ \epsilon^{*}_t, t = 1,...,n  \right\}$ are i.i.d with distribution function $\tilde{F}_n$. This pseudo-series allow us to construct the bootstrap sample $\left\{ X_t^{*}, t = 1,...,n  \right\}$ recursively by the following expression 
\begin{align}
X_t^{*} = \hat{\beta}_n X_{t-1}^{*} + \epsilon^{*}_t, \ \ \ \ t = 1,...,n, 
\end{align} 
with $X_0 = 0$. The bootstrap least squares estimate is then given by 
\begin{align}
\hat{\beta}_n^{*} = \left( \sum_{t=1}^n X^{*2}_{t-1} \right)^{-1} \left( \sum_{t=1}^n X^{*}_t X^{*}_{t-1} \right)  
\end{align}

The main aim of this section (see, also \cite{bose1988edgeworth}), is to derive the limit distribution of $\hat{\beta}_n^{*}$ given $\left( X_1, X_2,... \right)$ and to show that it is the same as the limit distribution given in Theorem \eqref{theorem.Anderson}.  Before doing that, we present below a useful lemma for the purpose of proving the aforementioned result. 

\begin{lemma}
\label{lemma1}
For the autoregressive model with $\mathbf{E} \epsilon_t = 0$, Var $\epsilon_t = \sigma^2 < \infty$ and $| \beta | > 1$, we have
\begin{align}
\frac{1}{n} \sum_{t=1}^n \left( \tilde{ \epsilon}_t  - \epsilon_t \right) \overset{ a.s }{ \to } 0, \ \ \text{as} \ n \to \infty. 
\end{align}
Moreover, 
\begin{align}
\hat{\beta}_n \overset{ a.s }{ \to } \beta \ \ \text{as} \ n \to \infty. 
\end{align} 
\end{lemma}

In practise, we denote the empirical distribution of the centred residuals $\tilde{\epsilon}_t$ as $\tilde{F}_t$. Notice that the bootstrap innovations $\left\{ \epsilon_t^{*} \right\}$ are conditionally independent with common distribution $\tilde{F}_T$.


\subsection{Limit distribution of the bootstrap estimate.}

We now state the main result of this section. 

\begin{theorem}
Conditionally on $\left( X_1,..., X_n \right)$ as $n \to \infty$ we have, for $| \beta | > 1$, 
\begin{align}
\left( \hat{ \beta }_n^2  - 1 \right)^{-1} | \hat{ \beta }_n |^{n} \left(  \hat{ \beta }_n^{*} - \hat{ \beta }_n \right) \to_{d} V \big/ U
\end{align}
for almost all sample paths $\left( X_1, X_2, ... \right)$, where $U$ and $V$ are defined in \eqref{defintion}, and $\left( \hat{\beta}_n , \hat{ \beta}_n^{*} \right)$ are defined, above. 
\end{theorem}

\begin{proof}
Define, 
\begin{align}
U_n^{*} = \sum_{t=1}^n \hat{\beta}^{ - (t - 1 )} \epsilon_t^{*} \ \ \text{and} \ \ V_n^{*} = \sum_{t=1}^n \hat{\beta}^{ - (n - t )} \epsilon_t^{*}
\end{align}

The result in the theorem can be deduced analogously to that of Theorem \eqref{theorem.Anderson} provided we show that, conditionally on $\left(  X_1, ... X_n  \right)$, $\left( U_n^{*}, V_n^{*} \right) \to_{d} \left( U, V \right)$ for almost all sample paths, where $U_n^{*}$, $V_n^{*}$, $U$ and $V$ are as defined above. This is due to the fact that the limit distributions of $\hat{ \beta}_n^{*}$ and $\hat{\beta}_n$ are determined respectively by those of $\left( U_n^{*}, V_n^{*} \right)$ and $\left( U, V \right)$. Thus, for the remaining derivations we shall show that the conditionally characteristic function of $\left( U_n^{*}, V_n^{*} \right)$ given $\left(  X_1, ... X_n  \right)$ converges to the characteristic function of $\left( U, V \right)$, for almost all sample paths $\left(  X_1, X_2, ... \right)$. Let $\mathbf{E}^{*} (.)$ denote the expectation with respect to the distribution $\tilde{F}_n$ (of $\epsilon_t^{*}$ ) conditional on $\left(  X_1, ... X_n  \right)$. Let $\phi_{ U_n^{*} (\tau)} = \mathbf{E}^{*} \left( \text{exp} \ i\tau U_n^{*} \right)$.  

Thus, it is shown that 
\begin{align}
\phi_{ U_n^{*} (\tau)} = \prod_{t=1}^n \mathbf{E}^{*} \left( \text{exp} \ i \tau  \hat{\beta}_n^{-(t-1)}  \epsilon_t^{*} \right)
\end{align}
converges a.s to $\phi_{ U ( \tau )}$ for each $\tau \in \mathbb{R}$.

To accomplish this it suffices to show that
\begin{align}
\label{condition}
\underset{ m \to \infty }{ \text{lim} } \ \underset{ n \leq m }{ \text{sup} } \sum_{t = m}^n \bigg| \mathbf{E}^{*} \left( \text{exp} \ i \tau  \hat{\beta}_n^{-(t-1)}  \epsilon_t^{*} \right) - 1  \bigg| = 0 \ \ \ \text{a.s},
\end{align}
where the convergence in \eqref{condition} is uniform on compact sets in $\mathbb{R}$ and it holds that
\begin{align}
\label{expression0}
\underset{ n \to \infty }{ \text{lim} } \bigg| \mathbf{E}^{*} \left( \text{exp} \ i \tau \hat{\beta}_n^{-(t-1)}  \epsilon_t^{*} \right) - \phi_{ \epsilon_1 } \left( \tau \beta^{-(t-1)}  \right)  \bigg| = 0 \ \ \ \text{a.s}.
\end{align}
for each $t \geq 1$ and uniformly in $\tau$ in bounded intervals.


First, note that 
\begin{align}
\bigg| \mathbf{E}^{*} \left( \text{exp} \ i \tau \hat{\beta}_n^{-(t-1)}  \epsilon_t^{*} \right) - 1  \bigg| 
&\leq 
\mathbf{E}^{*} \left(  |\tau \hat{ \beta}_n^{(t-1)} \epsilon_t^{*} | \right)
\nonumber
\\
&= | \tau | | \hat{ \beta}_n |^{-(t-1)} \frac{1}{n} \sum_{j=1}^n |  \tilde{\epsilon}_j^{*}  |. 
\end{align}
Since by Lemma \eqref{lemma1} $\hat{\beta}_n \overset{ \text{a.s} }{ \to } \beta$, $| \beta | > 1$ and $( 1/ n ) \sum_{j=1}^n | \tilde{\epsilon}_t | \overset{ \text{a.s} }{ \to } \mathbb{E} | \epsilon_1 | < \infty$, it follows that almost for almost each sample path there exist positive integers $C$ and $m$ such that $( 1/ n ) \sum_{j=1}^n | \tilde{\epsilon}_j | \leq C$ and $| \hat{\beta}_n | \geq \delta > 1$ for all $n \geq m$. 
Thus, for almost each sample path 
\begin{align}
\bigg| \mathbf{E}^{*} \left( \text{exp} \ i \tau \hat{\beta}_n^{-(t-1)}  \epsilon_t^{*} \right) - 1  \bigg| \leq | \tau | \delta^{- (t-1)} C
\end{align}
for all $n \geq m$, and \label{condition} is established. Next, we have that
\begin{align}
\label{expression1}
\bigg| \mathbf{E}^{*} &\left( \text{exp} \ i \tau \hat{\beta}_n^{-(t-1)}  \epsilon_t^{*} \right) - \phi_{ \epsilon_1 } \left( \tau \beta^{-(t-1)}  \right)  \bigg| 
\nonumber
\\
&\leq 
\bigg| \frac{1}{n} \sum_{j=1}^n \text{exp} \ i \tau \hat{\beta}_n^{-(t-1)} \tilde{\epsilon}_j - \frac{1}{n} \sum_{j=1}^n \text{exp} \ i \tau \hat{\beta}_n^{-(t-1)} \epsilon_j \bigg|
\nonumber
\\
&+
\bigg| \frac{1}{n} \sum_{j=1}^n \text{exp} \ i \tau \hat{\beta}_n^{-(t-1)} \epsilon_j - \frac{1}{n} \sum_{j=1}^n \text{exp} \ i\tau \beta_n^{-(t-1)} \epsilon_j \bigg|
\nonumber
\\
&+
\bigg| \frac{1}{n} \sum_{j=1}^n \text{exp} \ i\tau \beta_n^{-(t-1)} \epsilon_j - \mathbf{E} \left( \text{exp} \ i\tau \beta^{-(t-1)}  \epsilon_1 \right) \bigg|.
\end{align}
Therefore, for the first term in \eqref{expression1} we obtain,
\begin{align*}
\bigg| \frac{1}{n} &\sum_{j=1}^n \text{exp} \ i\tau \hat{\beta}_n^{-(t-1)} \tilde{\epsilon}_j - \frac{1}{n} \sum_{j=1}^n \text{exp} \ i\tau \hat{\beta}_n^{-(t-1)} \epsilon_j \bigg|
\\
&\leq \frac{1}{n} \sum_{j=1}^n \big| \text{exp} \ i\tau \hat{\beta}_n^{-(t-1)} \epsilon_j \big| \big| \text{exp} \ i \tau \hat{\beta}_n^{-(t-1)} \left( \tilde{\epsilon}_j - \epsilon_j  \right) - 1 \big|
\\
&\leq \frac{| \tau |}{n} \sum_{j=1}^n \big| \hat{\beta}_n^{-(t-1)} \big| \big| \tilde{\epsilon}_j - \epsilon_j \big|
\\
&\leq \left( |\tau| \hat{\beta}_n^{-2(t-1)} \right)^{1 / 2} \left( \frac{|\tau| }{n} \sum_{j=1}^n \left( \tilde{\epsilon}_j - \epsilon_j \right)^2 \right)^{1 / 2} \to 0 \ \ \ \text{a.s}.
\end{align*}
uniformly on $\tau$ on compact subsets by Lemma \ref{lemma1}. Similarly, we obtain that 
\begin{align}
\big| \hat{\beta}_n^{-(t-1)} - \beta^{-(t-1)} \big| \frac{| \tau |}{n} \sum_{j=1}^n | \epsilon_j | \to 0 \ \ \ \text{a.s}.
\end{align}
uniformly in $\tau$ on compact subsets, which implies that the second term in  \eqref{expression1} goes to zero almost surely.


By the Glivenko-Cantelli theorem and the convergence theorem for characteristic functions, the third term in \eqref{expression1} goes to 0 a.s. and uniformly for $\tau$ in a compact set. Hence, \eqref{expression0} is established, and it follows that 
\begin{align}
\phi_{U^{*}_n} ( \tau ) = \prod_{t=1}^n \phi_{\epsilon^{*}_1} \left( \tau \hat{\beta}_n^{-(t-1)} \right)  \overset{ \text{a.s} }{ \to } \prod_{t=1}^{\infty } \phi_{\epsilon_1 } \left( \tau \beta^{-(t-1)}  \right) = \phi_U ( \tau ).
\end{align}
Similarly, it can be shown that 
\begin{align}
\phi_{U^{*}_n, V^{*}_n} ( \tau, s ) \overset{ \text{a.s} }{ \to } \prod_{t=1}^{\infty }  \phi_{\epsilon^{*}_1} \left( \tau \hat{\beta}^{-(t-1)} \right)  \prod_{t=1}^{\infty }  \phi_{\epsilon^{*}_1} \left( s \hat{\beta}^{-(t-1)} \right) = \phi_U ( \tau ) \phi_V ( s ).
\end{align}
\end{proof}

\begin{corollary}
If $\left\{ \epsilon_t \right\}$ are assumed normal, we have, for $| \beta | > 1$, conditionally on $\left( X_1,..., X_n \right)$, 
\begin{align}
T_n = \hat{ \sigma}_n^{-1} \left( \sum_{t=1}^n X^{*2}_{t-1} \right)^{1 / 2} \left( \hat{\beta}_n^{*} - \hat{\beta}_n \right) \to_d \mathcal{N} \left( 0,1 \right),
\end{align}
for almost all sample paths, where
\begin{align}
\hat{ \sigma}_n^2 = \sum_{t=1}^n \left( X_t - \hat{\beta}_n X_{t-1} \right)^2.
\end{align}
\end{corollary}

\newpage

\section{Residual based Block Bootstrap for unit root testing.}
\label{appD}
\subsection{Estimation framework.}

Consider the autoregressive model given by 
\begin{align}
X_t = \rho_n X_{t-1} + U_t, \ \ \ t = 1,2,...
\end{align}
where $\rho_n = 1 + c / n$, $c < 0$ and the stationary process $\left\{ U_t  \right\}$ satisfies: $E \left( U_t \right) = 0$, $E \left( | U_t |^{\nu} \right) < \infty$ for some $\nu > 2$, $f_U(0) > 0$ and $\sum_{k=1}^{\infty} \alpha (k)^{1-2 / \nu } < \infty$, where $\alpha(.)$ denotes the strong mixing coefficient of $\left\{ U_t \right\}$. In this Section, we summarize the RBB testing Algorithm as proposed by \cite{paparoditis2003residual}. More specifically, the algorithm is carried out conditionally on the original data $\left\{ X_1,....,X_n \right\}$ and implicitly defines a bootstrap probability mechanism denoted by $P^{*}$ that is capable of generating bootstrap pseudo-series of the type $\left\{ X_t^{*}, t = 1,2,...  \right\}$.  

\subsection{RBB Testing Algorithm.}

The RBB procedure could be indeed an interesting application to examine, but we first consider the validity of the bootstrap IVX estimate constructed via the simple bootstrap procedure for obtaining the residual sequence. Then after establishing this asymptotic equivalence, we could examine the implementation of the RBB algorithm for obtaining the IVX bootstrap estimate.

\begin{enumerate}
\item[\text{Step 1.}] First calculate the centered residuals given by 
\begin{align}
\hat{v}_t = \left( X_t - \tilde{\rho}_n X_{t-1} \right) - \frac{1}{n-1} \sum_{j = 2}^n \big( X_j - \tilde{\rho}_n X_{j-1} \big)
\end{align}
for $t = 2,3,...,n $ where $\tilde{\rho}_n = \tilde{\rho}_n \big( X_1, X_2, ..., X_n \big)$ is a consistent estimator of $\rho_n$ based on the observations $\left\{ X_1,..., X_n \right\}$. 

\item[\text{Step 2.}] Choose a positive integer $b$, where $b < n$ and let $i_0,...,i_{k-1}$ be drawn i.i.d with distribution uniform on the set $\left\{1,2,..., n - b \right\}$, such that for example, $k = \left[ \frac{(n - 1)}{b} \right]$. Then, the procedure constructs a bootstrap pseudo-series $\left\{ X_1^{*},...., X_l^{*} \right\}$ where $l = kb + 1$, as follows
\begin{equation}
X_t^{*}
=
\begin{cases}
X_1    & , \text{for} \ t=1,
\\
\hat{\mu} + X^{*}_{t-1} + \hat{v}_{i_{m} + s} & , \text{for} \ t = 2,3,...,l,
\end{cases}
\end{equation} 
where $m = \left[ \frac{(t-2)}{b} \right]$, $s = t - mb - 1$, and $\hat{\mu}$ is a drift parameter that is either equal to zero or represents a consistent estimator of $\mu$. 

\item[\text{Step 3.}] Let $\hat{\rho}_n$ be the estimator used to perform the unit root test. Compute the pseudo-statistic $\rho^{*}$ which is the test statistic $\hat{\rho}_l$ based on the pseudo-data $\left\{ X_1^{*},...., X_l^{*} \right\}$. 
\end{enumerate}
Taking into account the asymptotic theory of the regression statistic $n \left( \hat{\rho}_{LS} - 1 \right)$ for near integrated processes, the following theorem about the asymptotic local power behavior of the RBB based test can be established.
\begin{theorem}
Let 
\begin{align}
\mathcal{B}_{RBB,n} \left( \rho_n ; \alpha \right) \to \mathbb{P} \left( J \leq \mathcal{C}_{\alpha} - c \right)
\end{align}
convergence in probability, where $\mathcal{C}_{\alpha}$ is the $\alpha$ quantile of the distribution of 
\begin{align}
\bigg( W^2(1) - \sigma_U^2 / \sigma \bigg) \bigg( 2 \int_0^1 W^2 (r) dr  \bigg)^{-1} 
\end{align} 
where J is a random variable which has the following distribution
\begin{align}
\bigg( \int_0^1 J_c (r) dW(r) + \left( 1 - \frac{ \sigma^2_U }{ \sigma^2 } \right) \bigg)\bigg( \int_0^1 J_c^2 (r) dr  \bigg)^{-1} 
\end{align}
and $J_c(r) = \int_0^1 e^{(r-s)c} dW(s)$ is the OU process generated by the stochastic differential equation $dJ_c(r) = cJ_c(r) dr + dW(r)$ with initial condition $J_c(0) = 0$.
\end{theorem}

\subsection{FCLT for the Bootstrap Partial Sum process.}
\label{FCLT.bootstrap.Appendix}

The asymptotic properties of the RBB testing procedure are based on the stochastic behaviour of the standardized partial sum process $\left\{ S_{m}^{*}(r), 0 \leq r \leq 1 \right\}$, defined by 
\begin{align}
\label{process1}
S_{m}^{*}(r) = \frac{1}{m} \sum_{t=1}^{j-1} \frac{v_t^{*}}{ \sigma^{*} }, \ \ \ \ \ \text{for} \ \ \frac{(j-1)}{m} \leq r \leq \frac{j}{m}, \ \ \ j = 2,...,m.  
\end{align} 
and also 
\begin{align}
\label{process2}
S_{m}^{*}(1) = \frac{1}{m} \sum_{t=1}^{m} \frac{v_t^{*}}{ \sigma^{*} }, \ \ \ \ \ \text{for} \ \ \frac{(j-1)}{m} \leq r \leq \frac{j}{m}, \ \ \ j = 2,...,m.  
\end{align} 
where for example, we define that $v_1^{*} \equiv X_1$, $v_t^{*} = X_t^{*} - \hat{\beta} - X_{t-1}^{*}$, for $t = 2,3,...,m$, and $\sigma^{*2} = \text{var} \left( m^{- 1 / 2} \sum_{ j = 1}^m v_j^{*} \right)$. 

Note that, the partial sum process $S_{m}^{*}(r)$ is considered to be a random element in the function space $\mathcal{D}[0,1]$, which represent the space of all real valued functions on the interval $[0,1]$ that are right continuous at each point and have finite left limits.  

The following theorem, given by \cite{paparoditis2003residual}, shows that under a general set of assumptions on the process $\left\{ X_t \right\}$, and conditionally on the observed series $X_1,X_2,..., X_n$, the bootstrap partial sum process defined by \eqref{process1} and \eqref{process2} converges weakly to the standard Wiener process on $[0,1]$. Moreover, we denote with $T_n^{*} = T_n^{*} \left( X_1^{*}, X_2^{*},...., X_n^{*} \right)$ is a random sequence based on the bootstrap sample $X_1^{*}, X_2^{*},...., X_n^{*}$ and $G$ is a random measure, then the we use $T_n^{*} \Rightarrow G$ to denote the convergence in probability, which means that the distance between the law of $T_n^{*}$ and the law of $G$ tends to zero in probability for any distance metrizing weak convergence. 

\newpage 

\begin{theorem}
\label{theorem.fclt}
Let $\left\{ X_t \right\}$ be a stochastic process, and assume that the process $\left\{ v_t \right\}$ defined by $v_t = X_t - \rho X_{t-1}$ satisfies the above regulatory conditions, and let $\tilde{\rho}_n$ be an estimator of $\rho$. Moreover, if $b \to \infty$ such that $b \big/ \sqrt{n} \to 0$ as $n \to \infty$, then
\begin{align}
S_{m}^{*} \Rightarrow W \ \ \ \text{in probability}. 
\end{align} 
\end{theorem}
Therefore, the above result along with a bootstrap version of the continuous mapping theorem enables us to apply the block bootstrap proposal of this paper in order to approximate the null distribution of a variety of different test statistics. 
\begin{theorem}
Assume that the process $\left\{ X_t \right\}$ satisfies the above conditions with $\beta = 0$. If $b \to \infty$ but $b \big/ \sqrt{n}$ as $n \to \infty$, then we have that 
\begin{align}
\underset{ x \in \mathbb{R} }{ \text{sup} } \ \bigg| \mathbb{P}^{*} \left( l \left(  \hat{ \rho }^{*LS}_n - 1 \right) \leq x \big| X_1,...,X_n \right) - \mathbb{P}_0 \left( \left( \hat{\rho}^{LS} - 1 \right) \leq x \right) \bigg| \to 0
\end{align}
in probability, and 
\begin{align}
\underset{ x \in \mathbb{R} }{ \text{sup} } \ \bigg| \mathbb{P}^{*} \left( l \left(  \hat{ \rho }^{*LS}_{c,n} - 1 \right) \leq x \big| X_1,...,X_n \right) - \mathbb{P}_0 \left( \left( \hat{\rho}_{c,n}^{LS} - 1 \right) \leq x \right) \bigg| \to 0
\end{align}
in probability, where $\mathbb{P}_{0}$ denotes the probability measure corresponding to the case where the statistics $\hat{ \rho }^{*LS}_{n}$ and $\hat{ \rho }^{*LS}_{c,n}$ are computed from a stretch of size $n$ from the unit root process obtained by integrating $\left\{ U_t \right\}$.
\end{theorem}

\section{Supplementary Limit Theorems.}

In this section we explain in more details the implications of having a random limit distribution in the related theory for bootstrapping. In other words, the theoretical result which we aim to prove involves inference based on stochastic limit bootstrap measures\footnote{Notice that our approach in this paper is different from examining the conditional versus unconditional validity of the bootstrap for predictive regression models. }. To begin with, we assume that the standard conditional weak convergence result applies, $\underset{ x \in \mathbb{R} }{ \text{sup} } \big| F_n^{*} (x) - F(x) \big| \to_p 0$. Furthermore, denote with $\mathbb{E}^{*}$ the expectation under the probability measure induced by the standard bootstrap. 
\begin{theorem}
Suppose that $\left\{ X_{n,j}, \mathcal{F}_{n,j} \right\}$ is a martingale difference array. Let $\left\{ \mathcal{J}_n(r),  r \in [0,1] \right\}$ be a sequence of adapted time scales and $\left\{ \mathcal{J}(r), r \in [0,1] \right\}$ a continuous, nonrandom function. If the following holds, 
\begin{align}
\forall \epsilon > 0, \ \sum_{j=1}^{ \mathcal{J}(1) } \mathbf{E} \left( X_{n,j}^2 \mathbf{1} \left\{ | X_{n,j} | > \epsilon  \right\} \big| \mathcal{F}_{n, j-1} \right) &\to_p 0, \ \ \text{as} \ \ n \to \infty
\\
\sum_{j = 1}^{ \mathcal{J}_n(r) } \mathbf{E} \left( X^2_{n,j} \big| \mathcal{F}_{n,j-1} \right) &\to_p  \mathcal{J}(r), \ \ \text{as} \ \ n \to \infty, r \in [0,1],
\end{align}

\newpage 

Then
\begin{align}
\sum_{j = 1}^{ \mathcal{J}_n(r) } X_{n,j} \to_d W \left(  \mathcal{J}(r)   \right), \ \ \text{as} \ \ n \to \infty, \ \ \text{in} \ \ \mathcal{D}[0,1].
\end{align}
\end{theorem}
A sequence $\left\{ P, P_n \right\}$ of probability measures on the metric space $( S, d )$ converges weakly, when 
\begin{align}
\int \phi(x) dP_n(x) \to \int \phi(x) dP(x), \ \ n \to \infty,  
\end{align}
holds true for all $\phi \in C_b ( S, \mathbb{R} )$. 

\medskip

\begin{theorem}
(Continuous Mapping Theorem) Let $\left\{ X , X_n \right\}$ be a sequence of random elements taking values in some metric space $( S, d )$ equipped with the associated Borel $\sigma-$field. Assume that 
\begin{align}
X_n \Rightarrow X, \ \ \text{as} \ \ n \to \infty, 
\end{align}  
\end{theorem}

If $\phi: S \to S^{\prime}$, is a mapping into another metric space $S^{\prime}$ with metric $d^{\prime}$ that is \textit{almost surely}, continuous on $X ( \Omega ) \subset S$, then 
\begin{align}
\phi ( X_n ) \overset{ d }{ \to } \phi(X), \ \ \text{as} \ \ n \to \infty,
\end{align} 

\begin{theorem}
(Joint Weak Convergence)
\
 
Let $\left\{ X, X_n \right\}$ and $\left\{ Y, Y_n \right\}$ be two sequences taking values in $( S_1, d_1 )$, respectively, $( S_2, d_2 )$, such that some conditions hold. Then, 
\begin{align}
( X_n, Y_n ) \Rightarrow ( X, Y),
\end{align} 
as $n \to \infty$, provided at least one of the following conditions is satisfied
\begin{enumerate}
\item[(i)] $Y = c \in S_2$ is a constant, that is, non-random. 

\item[(ii)] $X_n$ and $Y_n$ are independent for all $n$ as well as $X$ and $Y$ are independent. 
\end{enumerate}
\end{theorem}
Recall that we are given a metric space $(S, d)$ such that as $\mathcal{D} \left( [0,1], \mathbb{R} \right)$ equipped with the Borel $\sigma-$field and the associated set $\mathcal{P}(S)$ of probability measures. Then, the space $\mathcal{P}(S)$ can be metrized by the Prohorov metric $\pi$ and the convergence with respect to the Prohorov metric is the weak convergence, that is, 
\begin{align}
P_n \Rightarrow P \ \ \text{if and only if} \ \phi \left( P_n, P \right) \to 0. 
\end{align}

In a metric space convergence can be characterized by a sequence criterion: A sequence converges if and only if any subsequence contains a further convergent subsequence. Therefore, the weak convergence of a sequence $\left\{ P_n, P \right\}$ of probability measures can be characterized in the following way: 
\begin{align}
P_n \Rightarrow P \iff \pi ( P_n, P ) \to 0, \ \ \text{as} \ n \to \infty 
\end{align}  
if and only if any subsequence $\left\{ P_{n_k} : k \geq 1 \right\}$
contains a further subsequence $\left\{ P_{n_{k^{\prime}} } : k \geq 1 \right\}$ such that $P_{n_{k^{\prime}} } \Rightarrow P \iff \pi (  P_{n_{k^{\prime}} }, P ) \to 0$, as $k \to \infty$, which is equivalent to $P_{n_{k^{\prime}} } \Rightarrow P$, as $k \to \infty$. 

Further, in any metric space a subset $A$ is relatively compact, that is, has a compact closure, if every subsequence $\left\{ P_n \right\} \subset A$ has a subsequence $\left\{ P_{n_{k^{\prime}} } : k \geq 1    \right\}$ with $P_{n_{k^{\prime}} } \to P^{\prime}$ as $k \to \infty$, where the limit $P^{\prime}$ is in the closure of $A$ is relatively compact. Applied to out setting this means: A subset $A \subset \mathcal{P} ( \mathcal{S} )$ of probability measures has compact closure $\bar{A}$ if and only if every sequence $\left\{ P_n \right\} \subset A$ has a subsequence $\left\{ P_{ n_k } \right\}$ with converges weakly to some $P^{\prime} \in \bar{A}$, such that, $P_{ n_k } \Rightarrow P^{\prime}$ as $k \to \infty$. 

Here the limit $P^{\prime}$ may depend on the subsequence. Therefore, the weak convergence $P_n \Rightarrow P$ as $n \to \infty$. First, one shows that $\left\{ P_n \right\}$ is relatively compact and then one verifies that all possible limits are equal to $P$. A theorem due to Prohorov allows us to relate the compact sets of $\mathcal{P} (S)$ to the compact sets of $S$. This is achieved by the concept of tightness. A subset $A \subset \mathcal{P} (S)$ of probability measures in $\mathcal{P} (S)$ is called tight if for all $\epsilon > 0$ there exists a compact subset $K_{\epsilon} \subset S$ such that 
\begin{align}
P \left( K_{ \epsilon } \right) > 1 - \epsilon, \ \ \text{for all} \ \ P \in A.  
\end{align}

\paragraph{Functional Central Limit Theorems}

\begin{theorem}
(Donsker \textit{i.i.d} case) 

\
Let $\xi_1, \xi_2...$ be a sequence of $\textit{i.i.d}$ random variables with $\mathbb{E} \left( \xi_1 \right) = 0$ and $\sigma^2 = \mathbb{E} \left( \xi_1^2 \right) < \infty$. 
Then, 
\begin{align}
\frac{1}{ \sqrt{T} } \sum_{ t= 1}^{ \floor{ T s } } \xi_t \Rightarrow \sigma B (s), 
\end{align}
as $T \to \infty$, where $B$ denotes the standard Brownian motion and $\Rightarrow$ signifies weak convergence in the Skorohod space $\mathcal{D} \left( [0,1], \mathbb{R} \right)$. 
\end{theorem}

\begin{theorem}
Suppose $\xi_1, \xi_2,...$ satisfies a weak invariance principle. Then, 
\begin{align}
S_T(u) = \frac{1}{ \sqrt{T} } S \left( \floor{Tu} \right) \Rightarrow B(u), 
\end{align}
as $T \to \infty$. 
\end{theorem}

\begin{proof}
We have that $\left\{ B(u): u \geq 0 \right\}$ is equal in distribution to $\left\{ \frac{1}{ \sqrt{T} } B(Tu) : u \geq 0    \right\}$ for each $T$. Therefore, 
\begin{align}
\underset{ u \in [0,1] }{ \mathsf{sup} } \left| S_T(u) - B(u)   \right| \overset{ d }{ = } \underset{ u \in [0,1] }{ \mathsf{sup} } \left| \frac{1}{ \sqrt{T} } \sum_{t=1}^{ \floor{Tu} } \xi_t - \frac{1}{ \sqrt{T} } B(Tu) \right|. 
\end{align}
Therefore, we can conclude that on a new probability space, 
\begin{align*}
\underset{ u \in [0,1] }{ \mathsf{sup} } \left| S_T(u) - B(u)   \right| 
&\overset{ d }{ = } 
\underset{ u \in [0,1] }{ \mathsf{sup} } \frac{1}{ \sqrt{T} } \left| \sum_{t=1}^{ \floor{Tu} } \xi_t -  B(Tu) \right| 
\\
&=
\frac{1}{ \sqrt{T} } \ \underset{ n \leq T }{ \mathsf{max} } \ \left|  S(n) - B(n) \right| \overset{ p }{ \to } 0
\end{align*}
as $T \to \infty$, which implies that $S_T \Rightarrow B$ as $T \to \infty$, for the original processes. 
\end{proof}

Reference: Nuisance-parameter-free changepoint detection in non-stationary series

\medskip

\paragraph{Proof of Theorem 3}

Let $\left\{ X_t \right\}_{n=1}^{\infty}$ be i.i.d random variables. Then, the bootstrap partial sum process 
\begin{align}
S_n^{\star} (t) 
:= 
\frac{1}{ \sqrt{n} } V_n^{\star} \left( \floor{nt} \right) 
= 
\frac{1}{ \sqrt{n} } \sum_{k=1}^{\floor{nt} } Y_{n,k}^{\star}
\end{align} 
for $t \in [0,1]$ has conditionally on $\left\{ Y_{n,k}  \right\}_{k=1}^n$ the same distribution as $\big( W \left( \nu_n^{\star} (t) \right) \big)_{t \in [0,1]}$ for some standard Wiener process $W$ and
\begin{align}
\nu_n^{\star} (t) = \mathsf{Var} \left[ \frac{1}{ \sqrt{n} } \sum_{k=1}^{\floor{nt} } Y_{n,k}^{\star} \big| \left\{ Y_{n,k} \right\}_{k \leq n} \right].
\end{align} 
Then, if $\delta_n \to$ as $n \to \infty$, we have that $\nu_n^{\star} \to \nu$ uniformly in $t \in [0,1]$ almost surely.

\subsection{Additional Proofs}

Next following the excellent statistical framework of \cite{li2001bootstrapping}  we establish the convergence of finite-dimensional distributions and verify the tightness condition. In particular, for the finite-dimensional convergence, it is sufficient to show that for any finite $k$ and $0 < r_1 < ... < r_k < 1$, if 
\begin{align}
V_k 
&= 
\bigg[ S_{\floor{T r_1 }}, S_{\floor{T r_2 }},..., S_{\floor{T r_{k-1} }}, S_{\floor{T r_k}} \bigg]
\\
V^{*}_k 
&= 
\bigg[  S^{*}_{\floor{T r_1 }}, S^{*}_{\floor{T r_2 }},..., S^{*}_{\floor{T r_{k-1} }}, S^{*}_{\floor{T r_k}}  \bigg]
\end{align} 
then, it holds that $d_2 \left( V^{*}_k , V_k  \right) \to 0$. Following the same arguments as in the literature it can be shown that 
\begin{align}
d_2 \left( V^{*}_k , V_k  \right)^2 \leq d_2 \left( F_T , \widetilde{F}_T \right) \to 0. 
\end{align}
Furthermore, for tightness it suffices to show that there exists a non-decreasing function $\psi$ such that, for almost all sample paths, for $0 \leq r_1 \leq r \leq r_2 \leq 1$, 
\begin{align}
\mathbb{P}^{*} \bigg( \bigg\{ \bigg| S^{*}_{ \floor{T r} } - S^{*}_{ \floor{T r_1} } \bigg| \geq \lambda \bigg\} \bigcap \bigg\{ \bigg| S^{*}_{ \floor{T r_2} } - S^{*}_{ \floor{T r} } \bigg| \geq \lambda \bigg\} \bigg) 
\leq 
\frac{1}{ \lambda^4 } \bigg[ \psi(r_2) - \psi(r_1) \bigg] 
\end{align} 
Notice that $\epsilon_t^{*}$ are \textit{i.i.d} draws from $\left\{ \tilde{\epsilon}_t \right\}$. Then, by Markov's inequality we have that    
\begin{align*}
\mathbb{P}^{*} \bigg( \bigg\{ \bigg| S^{*}_{ \floor{T r} } - S^{*}_{ \floor{T r_1} } \bigg| \geq \lambda \bigg\} \bigcap \bigg\{ \bigg| S^{*}_{ \floor{T r_2} } - S^{*}_{ \floor{T r} } \bigg| \geq \lambda \bigg\} \bigg) 
\leq
\frac{ \mathbb{E}^{*} \bigg| S^{*}_{ \floor{T r} } - S^{*}_{ \floor{T r_1} } \bigg|^2 }{ \lambda^2 } \frac{ \mathbb{E}^{*} \bigg| S^{*}_{ \floor{T r_2} } - S^{*}_{ \floor{T r} } \bigg|^2 }{ \lambda^2 }
\end{align*}  
Notice that the above bound in probability holds because
\begin{align}
\left\{ \int_0^1 B_v(r)^2 dr \right\}^{1/2} \left\{ \int_0^1 B_v(r) dW_{\epsilon} (r) \right\} \equiv \mathcal{N} (0,1). 
\end{align}

\subsubsection{Unit Root bootstrap tests for AR(1) models}

Consider the autoregressive model
\begin{align}
X_t = \beta X_{t-1} + u_t, \ \ \ X_0 = 0,
\end{align}
Then, $\hat{\beta}_n$ is the OLS estimator of $\beta$ based on a sample of $n$ observations $( X_1,..., X_n )$. Although the OLS estimator $\hat{\beta}_n$  is consistent, its limit distribution is different for the three possible cases: stationary, unstable and explosive. In particular, for the unstable case such that $\beta = 1$, is is known that the risk of the normalized error of the estimator given by 
\begin{align}
K_n := \frac{1}{ \sigma_u } \left( \sum_{t=1}^n X_{t-1}^2 \right)^{1/2} \left( \hat{\beta}_n - 1 \right)
\end{align} 
converges weakly to 
\begin{align}
K := \frac{1}{2} \bigg[ W^2(1) - 1 \bigg] \left\{ \int_0^1 W^2 (t) dt \right\}^{- 1 / 2}, \ \ \ \text{as} \ \  n \to \infty 
\end{align}
where $W(t)$ is the standard Brownian motion on $[0,1]$. 

Therefore, in the particular paper the authors prove that their resampling algorithm is asymptotically correct under the null hypothesis, $H_0$, in the sense that it converges weakly to the limit distribution for almost all samples $( X_1,..., X_n )$.

\subsubsection{A bootstrap invariance principle}

The study of the asymptotic behaviour of the bootstrap least squares estimate relies on a bootstrap invariance principle, that is, a functional central limit theorem for a stochastic process constructed from the sequence of partial sums corresponding to the bootstrap resamples. 

Consider the sequence of partial sums 
\begin{align}
S_{n,0}^{*} = 0, \ \  S_{n,k}^{*} \sum_{j=1}^k \varepsilon_{n,j}^*, \ \ \ k = 1,...,n, \ n \in \mathbb{N}.
\end{align}
Then a sequence of continuous-time processes such that $\left\{ Y^*_n(s) : s \in [0,1], n \in \mathbb{N} \right\}$

\begin{itemize}

\item Thus, for the predictive regression model since the parameter of interest is the coefficient $\beta$ such that $H_0: \beta = 0$. Therefore, in other words in the spirit of Jeganthan (1980) we have proved that the nonstandard problem of testing for structural break in predictive regression models at an unknown location, has a limiting distribution with a discontinuity (a critical point) in the case of persistent predictors or more specifically when the exponent rate of persistence equals to one.   

\end{itemize}

Furthermore, we introduce some precise bootstrap terminology on bootstrap convergence that we need. Let $Y = \left( Y_1, Y_2,..., Y_n \right)$ be a random sample from a distribution $G$ and let $\uptau ( Y ; G )$ be the statistic of interest. A general goal of bootstrap resampling is to approximate the distribution function $\mathbb{P} \big( \uptau ( Y ; G ) \leq x \big)$ of the statistic  $\uptau ( Y ; G )$ by using the distribution function 
\begin{align}
\mathbb{P} \bigg( \uptau ( Y^{*} ; G ) \leq x \bigg)  \ \ \ \text{of} \ \ \ \uptau ( Y^{*} ; G ) , 
\end{align} 
where $Y^{*} = \left( Y^{*}_1, Y^{*}_2,..., Y^{*}_n \right)$ is the bootstrap resample.

Therefore, if $\uptau ( Y ; G )$ converges weakly to a distribution $S(G)$, it is enough to show that $\uptau ( Y^{*} ; G )$ converges weakly to $S(G)$ for almost all samples $(Y_1, Y_2, )$ or to establish that the distance between the law of $\uptau ( Y^{*} ; G )$ and the law of $S(G)$ tends to zero in probability for any distance metrising weak convergence. Therefore, our first lemma establishes the weak convergence of the finite dimensional distributions of the processes $Y^{*}_n (s)$ for almost all samples $( X_1,..., X_n )$. 

\medskip

\begin{lemma}
Conditionally on $( X_1,..., X_n )$ and for almost all sample paths we have that $( X_1, X_2, ... )$, $\big( Y^{*}_1(s_1), Y^{*}_2(s_2),..., Y^{*}_n(s_d) \big)$ converges weakly to $\big( W(s_1),..., W(s_d) \big)$ as $n \to \infty, \forall \  ( s_1,..., s_d ) \in [0,1]^d$. Furthermore, the tightness of our sequence of stochastic processes is implemented by the following results. 
\end{lemma}

Next, we establish the bootstrap invariance principle. 

\medskip

\begin{proposition}
Let $\left\{ \epsilon_n : n \in \mathbb{N} \right\}$ be a sequence of residuals. Define with $\widehat{F}_n$ to be the empirical distribution of $\hat{\epsilon}_t := \epsilon_t - \frac{ 1 }{ n } \sum_{j=1}^n \epsilon_t$, for $t = 1,...,n$ and let  $\hat{\epsilon}^{*}_{n,t}$ for $t = 1,...,n$ be independent random variables with distribution $\widehat{F}_n$. Furthermore, define with $\big\{ Y_n^{*}(s) : s \in [0,1] \big\}$ for $n \in \mathbb{N}$. Then, $Y_n^{*}$ converges weakly to $W$ in $C[0,1]$ almost surely as $n \to \infty$, where $W$ is the standard one-dimensional Brownian motion on $[0,1]$.  
\end{proposition}

\subsubsection{Asymptotic behaviour of the bootstrap statistic}

Define the following statistic
\begin{align}
Z_n^{*} := \frac{1}{ \widehat{\sigma}_n } \left( \sum_{t=1}^n X_{n, t-1}^{*2} \right)^{1 / 2} \left( \widehat{\beta}_n^{*} - 1   \right),
\end{align}
to be the bootstrap version of $Z_n$ under $\beta = 1$. We derive the limiting distribution of $Z_n^{*}$ in the theorem below. In order to prove the particular result we need to employ the following lemma. 

\bigskip

\begin{lemma}
Define with 
\begin{align}
r_n^{*} 
= 
\frac{1}{n} \sum_{i=1}^n Y_n^{*2} \left( \frac{i}{n} \right) - \int_0^1 Y_n^{*2} (s) ds.
\end{align}
Then $r_n^{*}$ converges in probability to $0$ as $n \to \infty$, conditionally on $\left( X_1,..., X_n \right)$, and for almost all sample paths $\left( X_1, X_2,... \right)$. 
\end{lemma}
Our goal is to show that $Z_n^{*}$ converges weakly to $Z$ as $n \to \infty$ almost surely and so this bootstrap resampling approaches properly the correct limiting distribution. 

\begin{theorem}
Under the AR(1) model with $\beta = 1$, $Z_n^{*}$ defined in (2.3) converges weakly to Z as $n \to \infty$ for almost all sample $\left( X_1, X_2, ..., \right)$, where $Z$ is defined in (1.2). 
\end{theorem}

Consider the following statistic 
\begin{align}
\mathcal{T}_n 
:= 
\frac{1}{ s_n } \left( \sum_{t=1}^n X_{n, t-1}^{2} \right)^{1 / 2} \left( \widehat{\beta}_n - 1   \right),
\end{align}
where 
\begin{align}
s_n^2 = \frac{1}{n} \sum_{t=1}^n \left( X_t - \widehat{\beta}_n X_{t-1} \right)^2. 
\end{align}

In particular, \cite{phillips1987time} showed that, under $H_0$, $\widehat{\beta}_n$ converges in probability to 1, as $n \to \infty$ and $\mathcal{T}_n$ converges weakly to 
\begin{align}
\mathcal{T} 
:= 
\frac{\sigma}{ 2 \sigma_u } \left\{ W^2(1) - \frac{ \sigma_u^2 }{ \sigma^2 } \right\} \left\{ \int_0^1 W^2(t) dt \right\}^{- 1 / 2}, \ \ \text{as} \ \ n \to \infty.
\end{align} 
with $\sigma_u^2 = \underset{ n \to \infty }{ \mathsf{lim} } \frac{1}{n} \sum_{t=1}^n \mathbb{E} \left( u_t^2 \right)$.

\newpage 

\subsubsection{Bootstrap Inference in Cointegrating Regressions}

Consider the following system studied by  \cite{psaradakis2001bootstrap} 
\begin{align}
y_t &= \boldsymbol{\beta}^{\top} \boldsymbol{x}_t  + u_t,
\\
\boldsymbol{x}_t &= \boldsymbol{x}_{t-1} + \boldsymbol{v}_t
\end{align}
where $\boldsymbol{\beta}$ is $d-$dimensional parameter vector and $\boldsymbol{\eta}_t = \left( u_t, \boldsymbol{v}_t^{\top}   \right)^{\top} \in \mathbb{R}^{d+1}$ is a strictly stationary and ergodic random process with zero mean, finite covariance matrix, and continuous spectral density matrix $f_{ \boldsymbol{\eta} }(.)$ which is positive definite at zero. For the above system we are interested in conducting inference about $\boldsymbol{\beta}$ based on the FM-OLS and CCR estimators. 

In order to define these estimators we consider the following matrices
\begin{align}
\boldsymbol{\Lambda} = \sum_{j=1}^{\infty} \mathbb{E} \left( \boldsymbol{u}_t \boldsymbol{u}_{t+j}^{\top} \right) \ \ \ \text{and} \ \ \ \boldsymbol{\Sigma} = \mathbb{E} \left(  \boldsymbol{u}_t \boldsymbol{u}_{t}^{\top} \right) 
\end{align}
and decompose the long-run covariance matrix $\boldsymbol{\Omega} = 2 \pi f_{ \boldsymbol{\eta} }(0)$ of the innovation sequence $\left\{  \boldsymbol{\eta} \right\}$ as
\begin{align}
\boldsymbol{\Omega} = \boldsymbol{\Lambda}^{\top} + \boldsymbol{\Delta}, \ \ \ \text{where} \ \ \  \boldsymbol{\Delta} =  \boldsymbol{\Lambda} + \boldsymbol{\Sigma} 
\end{align} 
such that 
\begin{align}
\boldsymbol{\Omega} =  
\begin{pmatrix}
\omega_{11} & \boldsymbol{\omega}_{21}^{\top}
\\
\boldsymbol{\omega}_{21} & \boldsymbol{\Omega}_{22}
\end{pmatrix}, \ \ \ \ \ \
\boldsymbol{\Delta} =  
\begin{pmatrix}
\delta_{11} & \boldsymbol{\delta}_{12}
\\
\boldsymbol{\delta}_{12}^{\top} & \boldsymbol{\Delta}_{22}
\end{pmatrix}, 
\ \ \ \
\boldsymbol{\Delta}_{2} = \left( \boldsymbol{\delta}_{12}^{\top}, \boldsymbol{\Delta}_{22}^{\top} \right)^{\top}
\end{align}
Then, the FM-OLS estimator is defined as below
\begin{align}
\tilde{\boldsymbol{\beta}} = \left( \sum_{t=1}^n \boldsymbol{x}_t  \boldsymbol{x}_t^{\top} \right) \left( \sum_{t=1}^n \boldsymbol{x}_t \tilde{y}_t - n \tilde{\boldsymbol{\delta}} \right),
\end{align}
where we have that 
\begin{align}
\tilde{y}_t 
&= 
y_t - \hat{\boldsymbol{\omega}}_{21}^{\top} \hat{\boldsymbol{\Omega}}^{-1}_{22} \Delta \boldsymbol{x}_t, \ \ \ \Delta \boldsymbol{x}_t = \big( \boldsymbol{x}_t - \boldsymbol{x}_{t-1} \big)
\\
\tilde{\boldsymbol{\delta}}
&= 
\left( \hat{\boldsymbol{\delta}}_{12}^{\top}, \hat{\boldsymbol{\Delta}}_{22}^{\top} \right) \left( 1 , -  \hat{\boldsymbol{\omega}}_{21}^{\top} \hat{\boldsymbol{\Omega}}^{-1}_{22} \right)^{\top} \equiv \hat{\boldsymbol{\delta}}_{12}^{\top} - \hat{\boldsymbol{\omega}}_{21}^{\top} \hat{\boldsymbol{\Omega}}^{-1}_{22} \hat{ \boldsymbol{\Delta}}_{22}^{\top}
\end{align}
where $\hat{\boldsymbol{\Delta}}$ and $\hat{\boldsymbol{\Omega}}$ are consistent estimators of $\boldsymbol{\Delta}$ and $\boldsymbol{\Omega}$ respectively. Therefore, testing hypotheses regarding the individual elements of $\boldsymbol{\beta}$ are based on fully modified t-type statistics which are constructed based on the estimator $\tilde{\boldsymbol{\beta}}$ and the diagonal elements of the covariance matrix given by 
\begin{align}
\tilde{\boldsymbol{V}} 
= 
\left( \hat{\omega}_{11} - \hat{\boldsymbol{\omega}}_{21}^{\top} \hat{\boldsymbol{\Omega}}_{22}^{-1} \hat{\boldsymbol{\omega}}_{21} \right) \left( \sum_{t=1}^n \boldsymbol{x}_t \boldsymbol{x}_t^{\top} \right)^{-1}.
\end{align}
Notice that the nonsingularity of $\boldsymbol{\Omega}$ is equivalent to
\begin{align}
\omega_{11.2} = \omega_{11} - \omega_{12} \boldsymbol{\Omega}_{22}^{-1} \omega_{21}
\end{align}


Next we consider the canonical cointegrating regression estimator of $\boldsymbol{\beta}$ proposed by \cite{park1992canonical}, defined as below
\begin{align}
\bar{\boldsymbol{\beta}}
= 
\left( \sum_{t=1}^n \bar{\boldsymbol{x}}_t \bar{\boldsymbol{x}}_t^{\top} \right)^{-1} \left( \sum_{t=1}^n \bar{\boldsymbol{x}}_t^{\top} \bar{y}_t \right),
\end{align}   
Therefore, in the sequel we take $\hat{\boldsymbol{\Sigma}} = n^{-1} \sum_{t=1}^n \hat{\boldsymbol{u}}_t \hat{\boldsymbol{u}}_t^{\top}$ and estimate the matrices $\boldsymbol{\Omega}$ and $\boldsymbol{\Delta}$ nonparametrically using Parzen lag window, a plug-in bandwidth estimator, and the least-squares residuals $\left\{ \hat{\boldsymbol{u}}_t \right\}_{t=1}^n$ prewhitened via a first-order autoregression.

\paragraph{Bootstrap Algorithm}

Our bootstrap methodology is based on the assumption that $\left\{ \boldsymbol{\eta}_t \right\}$ admits an AR$( \infty )$ representation such that 
\begin{align}
\boldsymbol{\eta}_t = \sum_{j=1}^{ \infty } \boldsymbol{\varphi}_j \boldsymbol{\eta}_{t-j} + \boldsymbol{\varepsilon}_t,
\end{align}
where $\left\{ \boldsymbol{\varepsilon}_t \right\}$ is a zero-mean white-noise process and an absolutely summable sequence of matrices $\big\{  \boldsymbol{\varphi}_j \big\}$ that satisfies the following regularity condition 
\begin{align}
\mathsf{det} \left( \boldsymbol{I}_p - \sum_{j=1}^{
\infty} \boldsymbol{\varphi}_j  z^j \right) \neq 0 \ \ \ \forall  |z | \leq 1.
\end{align}
The particular condition is an assumption that has been extensively used in the cointegration literature and which is satisfied by a large class of random processes, including a causal and invertible multivariate ARMA processes. Then, the idea is to approximate the above linear process representation of the innovation sequence of the predictive regression model by a finite-order autoregressive model for a suitably proxy for $\boldsymbol{u}_t$ and use this to generate bootstrap observations which is the common practise when implementing residual-based bootstrap schemes.  

Consider testing the following hypothesis using the bootstrap procedure
\begin{align}
\mathcal{H}_0: \beta_{i} = \beta_{i,0} \ \ \ \text{versus} \ \ \  \mathcal{H}_1: \beta_{i} \neq \beta_{i,0}
\end{align}

Then, the bootstrap procedure is described as below:


\begin{itemize}

\item[\textbf{Step 1.} ] Calculate the residuals such that 
\begin{align}
\boldsymbol{\eta}^{o}_t = \left( u_t^{o}, \Delta \boldsymbol{x}_{t}^{\top} \right)^{\top} \equiv \left(  y_t - \boldsymbol{\beta}^{o \top} \boldsymbol{x}_t , \Delta \boldsymbol{x}_{t}^{\top} \right)^{\top}, \ \ \text{for} \ \ t = 1,...,n 
\end{align}
where $\boldsymbol{\beta}^{o}$ denotes the FM-OLS or the CCR estimator of $\boldsymbol{\beta}$ obtained under the null hypothesis which implies that the $i-$th element of the parameter vector $\left[ \boldsymbol{\beta}^{o} \right]_i \equiv \beta_{i,0}$.

\item[\textbf{Step 2.} ]  For a fixed positive integer $p$, obtain the estimates $\left\{ \hat{\boldsymbol{\Phi}}_1,..., \hat{\boldsymbol{\Phi}}_p \right\}$ of the coefficient matrices of a $(d+1)-$variate AR$(p)$ model for the innovation sequence $\left\{ \boldsymbol{\eta}_t^{o} \right\}_{t=1}^n$. In particular, these may be Yule-Walker estimates, least-squares estimates or Burg-type estimates. Furthermore, define with $\left\{ \boldsymbol{e}_t \right\}_{t=p+1}^n$ to be the corresponding residuals.    

\item[\textbf{Step 3.} ] For a given positive integer $r$, take a random sample $\left\{ \boldsymbol{e}^{\star}_t \right\}_{t=-r+1}^n$ from the empirical distribution that puts mass $( n - p )^{-1}$ on each of the centered residuals such that   
\begin{align}
\bar{\hat{\boldsymbol{e}}}_t 
= 
\left( \hat{\boldsymbol{e}}_t - \frac{1}{n - p} \sum_{ t = p+1}^n \boldsymbol{e}_t \right), \ \ \ \text{for} \ t = p +1,...,n 
\end{align}
Then, construct a bootstrap noise series $\left\{ \boldsymbol{\eta}_t^{\star} \right\}_{t = - r + 1}^n$ via the recursion below
\begin{align}
\boldsymbol{\eta}_t^{\star} = \sum_{ j = 1}^p \hat{\boldsymbol{\Phi}}_j \boldsymbol{\eta}_{t-j}^{\star} + \boldsymbol{e}_t^{\star}, \ \ \ t = - r + 1,...,n,
\end{align}
where $\boldsymbol{\eta}_t^{\star} = \boldsymbol{0}$ for $t \leq - r$. 

\item[\textbf{Step 4.} ] Generate bootstrap replicates $\left\{ \boldsymbol{x}_t^{\star} \right\}_{t=1}^n$ and $\left\{ \boldsymbol{y}_t^{\star} \right\}_{t=1}^n$ by setting 
\begin{align}
\boldsymbol{x}_t^{\star} = \boldsymbol{x}_{t-1}^{\star} + \boldsymbol{v}_t^{\star} \ \ \ \text{and} \ \ \ y_t^{\star} = \boldsymbol{\beta}^{o \top} \boldsymbol{x}_t^{\star} + u_t^{\star},  \ \ \ \text{for} \ t = 1,...,n
\end{align}
where we have partitioned the innovation sequence $\boldsymbol{\eta}_t^{\star}$ such that $\boldsymbol{\eta}_t^{\star} = \left( u_t^{\star}, \boldsymbol{v}_t^{\star \top} \right)^{\top}$ and set $\boldsymbol{x}_0^{\star} = \boldsymbol{0}$.

\item[\textbf{Step 5.} ] Estimate $\boldsymbol{\beta}$ by FM-OLS or CCM using the bootstrap sample $\left\{ \left( y_t^{\star}, \boldsymbol{x}_t^{\star} \right) \right\}_{t=1}^n$ and compute the value of the associated $t-$statistic for $\mathcal{H}_0: \beta_i = \beta_{i,0}$, say $\mathcal{T}^{\star} ( \beta_{i,0} )$. 

\item[\textbf{Step 6.} ] Repeat steps 3-5 independently $B$ times to obtain sa sample $\big\{ \mathcal{T}_b^{\star} ( \beta_{i,0} ) \big\}_{b=1}^B$ of $\mathcal{T} ( \beta_{i,0} )$ values.

\item[\textbf{Step 7.} ] Approximate the null sampling distribution of  $\mathcal{T} ( \beta_{i,0} )$ by the empirical distribution of $\big\{ \mathcal{T}_b^{\star} ( \beta_{i,0} ) \big\}_{b=1}^B$. The bootstrap p-value of the observed realization $\hat{\mathcal{T}} ( \beta_{i,0} )$ of $\mathcal{T} ( \beta_{i,0} )$ is given by 
\begin{align}
p-value = \frac{1}{B} \sum_{ b = 1}^B \mathbf{1} \left\{ \big| \mathcal{T}_b^{\star} ( \beta_{i,0} ) \big|  \leq  \left| \hat{\mathcal{T}} ( \beta_{i,0} ) \right|  \right\} 
\end{align}

\end{itemize}

\medskip

\begin{remark}
The bootstrap procedure described above uses the residuals obtained under the null hypothesis based on the Studentized difference between $\beta_{i,0}$ and an estimate of $\beta_i$. In particular, bootstrap tests based on such a resampling scheme have generally been found to have good small-sample size and power properties.Furthermore, for the implementation of the sieve bootstrap a necessary step is the selection of an appropriate value for the order $p$ of the approximating autoregression for the innovation sequence $\boldsymbol{\eta}_t$.   
\end{remark}

\begin{remark}
Notice that we are particularly interested to obtain a good enough autoregressive approximation for tests to have the correct empirical Type I error probability. More precisely, as it has been shown in the literature, the closure with respect to certain metrics of the class of stationary linear processes and AR$( \infty )$ processes is fairly large. In particular, this implies that the sieve bootstrap is likely to work reasonably well even when the data-generating mechanism does not belong to the AR$( \infty )$ family.  
\end{remark}

\newpage

In other words, the empirical distribution of $n^{- 1 / 2} \mathcal{T}_b^{\star} ( \beta_{i,0} )$ yields a conservative estimate of the empirical distribution of $n^{- 1 / 2} \mathcal{T} ( \beta_{i,0} )$ as $n, B \to \infty$. Although, we do not consider formal cross-validation methods in this paper to evaluate the performance of the bootstrap procedure such methodologies are discussed in the paper of \cite{efron1983estimating}. Furthermore, our goal is to investigate whether our bootstrap implementation indeed leads to an asymptotically valid test. In particular, this can be seen by examining whether the bootstrap consistently estimates the desired quantile.    
 
\subsubsection{Unit root and Cointegrating limit theory }

Consider the cointegrating system as below
\begin{align}
\boldsymbol{y}_t &= \boldsymbol{A} x_t + \boldsymbol{u}_t,
\\
\boldsymbol{x}_t &= \boldsymbol{x}_{t-1} + \boldsymbol{v}_t
\end{align}
where $\boldsymbol{y}_t$ is an $m-$dimensional vector and $\boldsymbol{x}_t$ is $K-$dimensional such that $\boldsymbol{\eta}_t = \left( \boldsymbol{u}_t^{\prime}, \boldsymbol{v}_t^{\prime}  \right)^{\prime}$ is an $(m + K)-$dimensional vector of innovations and $\boldsymbol{A}$ is an $(m \times K)$ matrix of cointegrating coefficients.  

The mixture process of the limit theory is given by 
\begin{align}
\mathsf{vec} \left\{ n \left( \hat{ \boldsymbol{A} }^{+} - \boldsymbol{A}  \right) \right\} \Rightarrow \mathcal{MN} \left( \boldsymbol{0},  \left( \int_0^1 \boldsymbol{B}_{\uptau}^{+} \boldsymbol{B}_{\uptau}^{+ \prime} \right)^{-1} \otimes \boldsymbol{\Omega}_{yy.x}   \right),
\end{align}
where $\hat{A}^{+}$ is the FM regression estimator with a conditional long-run covariance matrix of $u_t$ given $v_t$ expressed as below
\begin{align}
\boldsymbol{\Omega}_{yy.x} := \boldsymbol{\Omega}_{yy} - \boldsymbol{\Omega}_{yx} \boldsymbol{\Omega}_{xx}^{-1} \boldsymbol{\Omega}_{xy}
\end{align}
Then, the FM regression estimator has the explicit form given by
\begin{align}
\hat{ \boldsymbol{A} }^{+} = \left( \hat{ \boldsymbol{Y} }^{+ \prime} \boldsymbol{X} - n \hat{ \boldsymbol{\Delta} }^{+}_{yx} \right) \times \left( \boldsymbol{X}^{\prime} \boldsymbol{X} \right)^{-1}
\end{align}
where 
\begin{align}
\boldsymbol{X} = \big[ x_1^{\prime},..., x_n^{\prime} \big], \ \ \    \hat{ \boldsymbol{Y} }^{+} = \big[ \hat{y}_1^{+ \prime},..., \hat{y}_n^{+ \prime} \big]^{\prime} \in \mathbb{R}^{ n \times m} 
\end{align}
such that 
\begin{align}
\hat{\boldsymbol{y}}_t^{+} = \boldsymbol{y}_t - \hat{ \boldsymbol{\Omega} }_{yx}  \hat{ \boldsymbol{\Omega} }_{xx}^{-1} \Delta \boldsymbol{x}_t,  \ \ \ \text{and} \ \ \ \boldsymbol{\Delta}_{yx}^{+} = \boldsymbol{\Delta}_{yx} - \boldsymbol{\Omega}_{yx} \boldsymbol{\Omega}_{xx}^{-1} \boldsymbol{\Delta}_{xx}      
\end{align}
where $\hat{ \boldsymbol{\Omega} }_{yx}  \hat{ \boldsymbol{\Omega} }_{xx}^{-1}$ and  $\hat{\boldsymbol{\Delta}}_{yx}^{+}$ are consistent estimates of $\boldsymbol{\Omega}_{yx}  \hat{ \boldsymbol{\Omega} }_{xx}^{-1}$ and  $\boldsymbol{\Delta}_{yx}^{+}$ respectively, which may be constructed in the familiar fashion using semiparametric lag kernel methods with residuals from a preliminary cointegrating least squares regression.

\newpage

\begin{proof}
Setting with $u_{y.xt} = u_{yt} - \Omega_{yx} \Omega_{xx}^{-1} \Delta x_t$ and $U_{y.x} = \big[ u_{y.x1}^{\prime},..., u_{y.xn}^{\prime} \big]^{\prime}$ as the corresponding data matrix, we have that 
\begin{align}
\hat{\boldsymbol{y}}_t^{+} 
= 
\boldsymbol{y}_t - \hat{ \boldsymbol{\Omega} }_{yx}  \hat{ \boldsymbol{\Omega} }_{xx}^{-1} \Delta \boldsymbol{x}_t 
=
\boldsymbol{A} \boldsymbol{x}_t + \boldsymbol{u}_{0.xt} 
\end{align}

\end{proof}

In terms of asymptotic theory we have that 
\begin{align}
\xi_{ y . xn}^{+} (s) := \xi_{ yn } (s) - \boldsymbol{\Omega}_{yx} \boldsymbol{\Omega}_{xx}^{-1} \xi_{ xn}(s) \Rightarrow \boldsymbol{B}_{y.x} (s) \equiv BM \big( \boldsymbol{\Omega}_{y.xx}  \big)
\end{align}

\begin{remark}
Notice that the unit root limit theory given by expressions (17) and (18) of PM (2009) involves the demeaned process $B^{\mu}(s)$ although there is no intercept in the regression. In particular, the demeaning effects arises because as shown in expression (16), in the direction of the initial condition, the time series is dominated by a component that behaves like a constant. 
\end{remark}
Therefore, the sample moment matrix is no longer asymptotically singular such that 
\begin{align}
\frac{1}{n^2} \sum_{t=1}^n \boldsymbol{x}_{t-1} \boldsymbol{x}_{t-1}^{\prime} \Rightarrow \int_0^1 J_c^{*} (r) J_c^{* \prime} (r)
\end{align}

Related theory regarding the weak convergence to the matrix stochastic integral is presented by \cite{phillips1988weak}. We briefly summarize the main results below. We define the partial sum process such that $S_t = \sum_{ j=1}^t u_j$ and construct the following random element of $D^n [0,1]$ such that 
\begin{align}
X_T (r) = T^{- 1 / 2} S_{ \floor{Tr} } = T^{- 1 / 2} S_{j-1}, \ \ \ \ \text{with} \ \ \frac{(j-1)}{T} \leq r \leq \frac{j}{T}.
\end{align}
Then as $T \to \infty$, we have that 
\begin{align}
X_T (r) \Rightarrow B(r), 
\end{align}
where $B(r)$ is vector Brownian motion on $C^n[0,1]$, with covariance matrix $\Omega$. Therefore, one major time series application is to the theory of regression for integrated processes. In particular, $\left\{ u_t \right\}_{ t = 1 }^{\infty}$ is generated by a linear process such as a finite-order stationary and invertible vector ARMA model then $y_t$ is known as an integrated process of order one. Furthermore, we are often interested in the asymptotic behaviour of statistics from linear least-squares regressions with integrated processes. Then, from the first-order vector autoregression of $y_t$ on $y_{t-1}$ we obtain the regression coefficient matrix such that 
\begin{align}
\hat{A} = \left( \sum_{t=1}^T y_t y_{t-1}^{\prime} \right) \left( \sum_{t=1}^T y_{t-1} y_{t-1}^{\prime} \right)^{-1}. 
\end{align} 

\newpage 

The asymptotic behaviour of $\hat{A}$ is described by a corresponding functional Brownian motion. To be more precise consider standardized deviations of $\hat{A}$ about $I_n$ such that
\begin{align}
T \left( \hat{A} - I_n \right) = \left( \frac{1}{T} \sum_{t=1}^T u_t y_{t-1}^{\prime} \right) \left( \frac{1}{T^2} \sum_{t=1}^T y_{t-1} y_{t-1}^{\prime} \right)^{-1}. 
\end{align}   
Therefore, to obtain the asymptotic behaviour of the above statistic we write the sample second moment $\left( \frac{1}{T^2} \sum_{t=1}^T y_{t-1} y_{t-1}^{\prime} \right)$ as a quadratic functional of the random element $X_T(r)$, at least up to a term of $o_p(1)$. That is, 
\begin{align}
\frac{1}{T^2} \sum_{t=1}^T y_{t-1} y_{t-1}^{\prime} = \int_0^1 X_T(r) X_T(r)^{\prime} dr + o_p(1)  
\end{align}
Then, by an application of the continuous mapping theorem we can establish that 
\begin{align}
\frac{1}{T^2} \sum_{t=1}^T y_{t-1} y_{t-1}^{\prime} 
\Rightarrow \int_0^1 B_{(r)} B_{(r)}^{\prime} dr, \ \ \ \text{as} \ \ T \to \infty.  
\end{align}
Therefore, in a similar way the following limit result holds
\begin{align}
\frac{1}{T} \sum_{t=1}^T u_t y_{t-1}^{\prime} \Rightarrow \int_0^1 dB(r) B(r)^{\prime} 
\end{align}

Furthermore, the paper of \cite{Phillips2010bootstrapping}\footnote{The particular monograph in the words of the author is: "\textit{Dedicated to Phoebus Dhrymes whose advanced textbooks in econometrics have trained and educated generations of econometricians and applied researchers}".}, examines the asymptotics of the  bootstrap implementation for integrated time series data. See also the framework presented by \cite{parker2006unit}. Although in this paper we focus on the implementation of our proposed bootstrap algorithm for nonstationary autoregressive processes in predictive regressions models, when these are assumed to follow local-to-unity processes, the asymptotic validity of the bootstrap predictability tests is based on the fact that these employ a partial sum process to generate the bootstrap samples under the null, thereby ensuring that the test conforms to a unit root limit distribution.  

In particular, \cite{Phillips2010bootstrapping}  consider the continuous moving block bootstrap for a unit root process such as $x_t = x_{t-1} + u_t$ where $t = 1,...,n$ where $u_t$ is a linear process, to obtain the bootstrap sample $\left\{ x^{cb*}_{t}  \right\}_{t=1}^n$. Then, in order to construct a unit root testing based on the continuous moving block bootstrap the limit distribution of the serial correlation coefficients is obtained such that 
\begin{align}
\widehat{\rho}^{cb *} =  \frac{ \displaystyle \frac{1}{ n^2 } \sum_{t=1}^n x_t^{cb *}  x_{t-1}^{cb *} }{ \displaystyle \frac{1}{ n^2 } \sum_{t=1}^n  x_{t-1}^{cb * 2} }    
\end{align}  

\newpage

In particular, the denominator can be shown to have the following limit
\begin{align}
\frac{1}{ n^2 } \sum_{t=1}^n  x_{t-1}^{cb * 2} \overset{ d^{*} }{ \to } \int_0^1 B(r)^2 dr, 
\end{align}
However, as it has been established by \cite{Phillips2010bootstrapping} although the continuous moving block bootstrap proposed by \cite{paparoditis2001unit} is consistent under the null hypothesis that the sample data is I(1), which can be used to construct estimates of the finite sample distribution of unit root tests that involve semiparametric corrections for serial correlation, the particular approach does not produce consistent estimates of the null unit root under the alternative hypothesis. Therefore, we focus on the bootstrap procedure of \cite{paparoditis2003residual} that use residual-based partial summation methods when constructing the bootstrap sample.

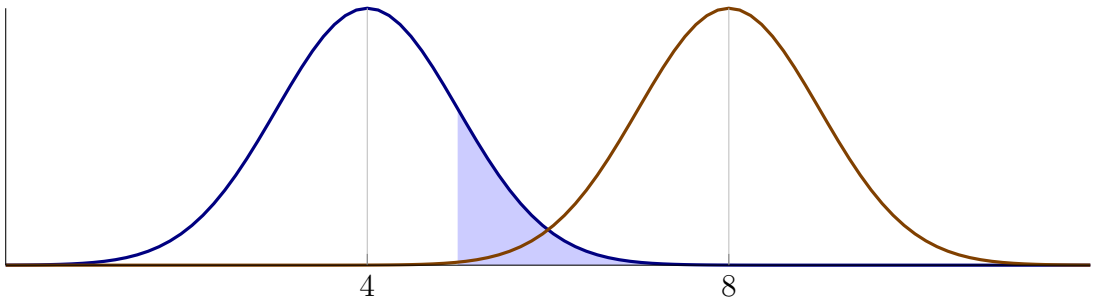
\pgfmathdeclarefunction{gauss}{2}{%
  \pgfmathparse{1/(#2*sqrt(2*pi))*exp(-((x-#1)^2)/(1.4*#2^2))}%
}

\begin{center}

\begin{tikzpicture}
\begin{axis}[
  no markers, domain=0:12, samples=100,
  axis lines*=left,  
  every axis y label/.style={at=(current axis.above origin),anchor=south},
  every axis x label/.style={at=(current axis.right of origin),anchor=west},
  height=5cm, width=16cm,
  xtick={4,8.0}, 
  ytick=\empty,
  enlargelimits=false, clip=false, axis on top,
  grid = major
  ]
  
  \addplot [fill=blue!20, draw=none, domain=5:7.96] {gauss(4,1.2)} \closedcycle;

  \addplot [very thick,blue!50!black] {gauss(4,1.2)};
  \addplot [very thick,orange!50!black] {gauss(8.0,1.2)};

\end{axis}
\label{fig}
\end{tikzpicture}

\end{center}

The above Figure illustrates the bootstrap-based test and the alternative hypothesis, such that $\mathbb{H}_0 : \beta \neq 0$. The power of the bootstrapped predictability test is determined by the mean and the variance of the limiting distribution under the alternative hypothesis. Therefore, if the mean of the limit distribution under the alternative hypothesis is sufficiently large, then the bootstrap-based test successfully rejects the null hypothesis with high probability.

\newpage 

\section{Efficient Tests}

Although not based on formal optimality criterion, comparing the various testing methodologies from the literature, the procedure of Ng and Perron (2001) has an asymptotic local power function that is indistinguishable from the near-effcients tests of \cite{elliott1996efficient} and superior power than the corresponding test based on the OLS demeaning considered in \cite{stock1994unit} and \cite{perron1991continuous}. One therefore conclude that it is largely the quasi-GLS method of demeaning that brings about this power advantage over the standard OLS demeaned unit root tests. 

Therefore, rather than focusing on the performance of predictive tests in the strongly dependent case that are driven by a formal asymptotic optimality property, we explore whether, and if so in what settings, using quasi-GLS demeaning of the persistent predictor can deliver statistical tests with good power. Indeed this is the case, as it is shown in Theorem 1 below, under strong persistence the limiting distribution of $\mathcal{T}$ features a component which is a weighted combination of two distributions, the first of which is the local alternative limit distribution of the OLS-demeaned Dickey-Fuller statistic and the second is a standard normal. Specifically, the Dickey-Fueller component dominates the standard normal component when the degree of endogeneity $\left| \rho_{xy} \right|$ is large. Where the degree of endogeneity is small the reverse holds and so here we might not necessarily expect to see any gains from using a test based on quasi-GLS dmeaning the predictor. Thus, an exploration of of the asymptotic local power functions of asymptotically conservative implementations - needed to account for the dependence on $c$ and $\rho_{xy}$ under the null - of the tests shows that the quasi-GLS demeaning of the persistent predictor can indeed deliver power gains relative to $\mathcal{T}$ for moderate to large $\rho_{xy}$ in the strongly persistent case. 

Furthermore, the authors discuss a t-ratio from a variant of the standard predictive regression (using a predictive regression model with a single predictor and a model intercept) where the OLS demeaned returns and regressed on the quasi-GLS demeaned lagged predictor. Specifically, in the case of strongly persistent predictor, we have proposed a feasible method for obtaining (conservative) asymptotic critical values of the test statistics. Moreover, an analysis of the asymptotic local power functions of the resulting (asymptotically) conservative tests in the case where the predictor is strongly persistent showed that these vary considerably with the endogeneity correlation parameter. The authors conclude that under the assumption of a common order of persistence, it may be possible to generalize their proposed statistical testing approach to accommodate multiple predictors. Lastly, the authors mention that investigating this possibility and how well it works in practice compared to the other tests mentioned above is beyond the scope of the present paper but would constitute an interesting topic for further research. 

Semiparametric efficiency has long been an essential issue in the econometric literature, where a common situation is to treat the innovation density as the nonparametric component. The theory for the standard locally asymptotically normal (LAN) or locally asymptotically mixed normal (LAMN) experiment has been well developed. The main idea follows the insight of Stein (1956) that the semiparametric efficiency bound is given by the lowest aming all the bounds of parametric submodels embedding the unknown true model, and the corresponding submodel is then called least-favorable.  

\newpage 

Unfortunately, this so-called least-favorable parametric submodel is not straightforwardly extended to some nonstandard limit such as the case of LABF experiments. In particular, the main challenge of the proposed framework is that the so-called least-favorable parametric submodel approach cannot be easily extended in the case of nonstandard statistical problem such as unit root testing and their asymptotic limits such as the locally asymptotically Brownian functional experiments. In particular, it is difficult to conduct statistical inference due to the fact that in those cases the least favorable submodel is not easily computationally tractable based on the conventional score-projecting method.

The particular challenge arises due to the inability of the conventional statistical theory to be directly applicable in this scenario. Specifically, by the Neyman-Pearson lemma, as any invariant test statistic must be a function of maximal invariant test function, the likelihood ratio of the maximal invariant defines the power envelope for all invariant tests. Consequently, the obtained maximal invariant likelihood ratio of the maximal invariant test statistic defines the power envelope for all invariant tests. 

Now in the case of testing methodologies in which the statistical procedure is allowed to employ a hybrid approach then it achieves a favourable test performance to the conventional method due to the fact that the testing methodology is allowed to change according to the degree of endogeneity, which is affecting the applicability of the maximal invariant principle. In other words, even though the current approach does not replace the particular statistical principle it manages to mimic its asymptotic behaviour thereby providing a good approximation of the true power envelope function under the alternative hypothesis. However, even in the novel approach of JM (Ecta, 2006) the authors mention that the proposed theory even though is developed based on related parameter restrictions and optimality criteria the optimality theory applies only in the case when the Gaussianity assumption is not violated which in a way is too restrictive since the applicability of the theory is based on the underline distributional assumption of the innovation sequence used for the data generating mechanism. More recently, the proposed statistical testing approach of MP (2022) can accommodate robust inference regardless of the assumption on the distribution of the innovation sequence, thereby allowing a distribution-free approach due to the nature of the proposed estimator that is robust across the spectrum of stationarity and nonstationarity.

\newpage 

\section{R Coding Procedures}

\subsection{Bootstrap for IVX Estimation under structural breaks (indicative)}

\begin{verbatim}

##################################################
# Step : Estimate bootstrap sup-Wald IVX statistic 
##################################################
n <- NROW(Xt)
y.t.star.matrix <- matrix( 0, nrow = n, ncol = B )
    
for ( b in 1:B )
{
model1     <- lm( Yt  ~ Xlag - 1) 
u.hat      <- as.matrix( as.vector( residuals( model1 ) ) )
kappa.t    <- as.matrix( rnorm( n, 0 , 1 ) )
u.hat.star <- u.hat * kappa.t
y.t.star.matrix[ ,b] <- as.matrix( u.hat.star ) 
}
    
bootstrap.Wald.IVX.matrix <- matrix( 0, nrow = B, ncol = 1 )

output <- foreach ( j = 1:10 , .combine = 'c' ) %do% 
{
    
y.t <- y.t.star.matrix[ ,j]  
estimation.Wald.IVX.function  <- estimation_Wald_IVX_function( Yt = y.t, Xt = Xt, 
Xlag = Xlag, pi0 = 0.15 ) 
sup.Wald.IVX.statistic        <- as.numeric( estimation.Wald.IVX.function )
bootstrap.Wald.IVX.matrix[j , 1] <- sup.Wald.IVX.statistic
    
return( bootstrap.Wald.IVX.matrix ) 
    
}# end dopar

bootstrap.IVX.statistic <- as.matrix( sort(  bootstrap.Wald.IVX.matrix) )
boot.IVX.statistic      <- bootstrap.IVX.statistic[(0.95*B),1]
sup_Wald_IVX_matrix_bootstrap[i ,1] <- boot.IVX.statistic 

if (  sup_Wald_IVX_statistic  > boot.IVX.statistic )
{
  empirical_size  <-  ( empirical_size + 1 ) 
}
    
\end{verbatim}

\end{appendix}

\newpage 

\bibliographystyle{apalike}

\footnotesize{
\bibliography{myreferences1}}

\end{document}